\documentclass{article} 
\usepackage{ijcai26}

\usepackage{times}
\usepackage{soul}

\usepackage[switch]{lineno}

\usepackage[utf8]{inputenc}
\usepackage{mathtools}
\usepackage{marvosym}
\usepackage{microtype}
\usepackage[labelfont=bf]{caption}
\usepackage{float}
\usepackage{amsfonts,amssymb,amsmath, amsthm}
\usepackage[ruled, noline, linesnumbered, commentsnumbered, noend]{algorithm2e}
\usepackage{thm-restate}
\usepackage[end]{algpseudocode}
\usepackage{xspace}
\usepackage[table,usenames,dvipsnames]{xcolor}
\usepackage[symbol]{footmisc}
\usepackage{booktabs, tabularx, makecell, multirow}
\usepackage{thmtools}
\usepackage[shortlabels]{enumitem}
\usepackage{graphicx}
\usepackage{subcaption}
\usepackage{todonotes}
\usepackage{etoolbox}
\usepackage{comment}
\usepackage{tikz}
\usepackage{hyperref}
\usepackage{url}

\definecolor{blueLink}{rgb}{0,0.2,0.8}
\hypersetup{colorlinks,linkcolor=blueLink,urlcolor=blueLink,citecolor=blueLink}

\usepackage{amsmath,amsfonts,bm}

\def\eqref#1{equation~\ref{#1}}

\def\1{\bm{1}}

\DeclareMathAlphabet{\mathsfit}{\encodingdefault}{\sfdefault}{m}{sl}
\SetMathAlphabet{\mathsfit}{bold}{\encodingdefault}{\sfdefault}{bx}{n}

\definecolor{arashcolor}{rgb}{0.0735, 0.49, 0.98}

\newtheorem{protocol}{Algorithm}
\newcommand{\sln}{\textsc{SLN}\xspace}
\newcommand{\csln}{\textsc{C-SLN}\xspace}
\newcommand{\uniform}{\textsc{Uniform}\xspace}
\newcommand{\laslin}{\textsc{LASLiN}\xspace}

\newcommand{\ossln}{\textsc{OSSLN}\xspace}

\newtheorem{theorem}{Theorem}
\newtheorem{definition}{Definition}
\newtheorem{lemma}{Lemma}

\newcommand{\olga}[1]{
  {\color{orange}(Olga: #1) \color{black}}
}

\title{
  A Learning-Augmented Overlay Network
}

\author{Julien Dallot, Caio Caldeira, Arash Pourdamghani, Olga Goussevskaia, Stefan Schmid \thanks{This paper has been supported by Austrian Science Fund (FWF), project I 5025-N (DELTA),
2020-2024.} \\
Technical University of Berlin, Universidade Federal de Minas Gerais
}

\begin{document}

\maketitle

\begin{abstract}
  This paper studies the integration of machine-learned advice in overlay networks in order to adapt their topology to the incoming demand.
  Such demand-aware systems have recently received much attention, for example in the context of data structures (Fu et al. in ICLR 2025, Zeynali et al. in ICML 2024).
  We in this paper extend this vision to overlay networks where requests are not to individual keys in a data structure but occur between communication pairs, and where algorithms have to be distributed. 
  In this setting, we present an algorithm that adapts the topology (and the routing paths) of the overlay network to minimize the hop distance travelled by bit, that is, distance times demand.
  In a distributed manner, each node receives an (untrusted) prediction of the future demand to help him choose its set of neighbors and its forwarding table.

  This paper focuses on optimizing the well-known \emph{skip list networks} (SLNs) for their simplicity.
  We start by introducing \emph{continuous skip list networks} (C-SLNs) which are a generalization of SLNs specifically designed to tolerate predictive errors.
  We then present our learning-augmented algorithm, called LASLiN, and prove that its performance is (i) similar to the best possible SLN in case of good predictions ($O(1)$-consistency) and (ii) at most a logarithmic factor away from a standard overlay network in case of arbitrarily wrong predictions ($O(\log^2 n)$-robustness, where $n$ is the number of nodes in the network).
  Finally, we demonstrate the resilience of LASLiN against predictive errors (ie, its smoothness) using various error types on both synthetic and real demands.
\end{abstract}

\section{Introduction}

Overlay networks are one of the key technological advancements that allowed decentralized and distributed systems to scale.
A vast body of work has been devoted to optimize the efficiency of these networks~\cite{LuaCPSL05,RatnasamyFHKS01,AspnesS03,HarveyJSTW03,KifferSLMN21}.
Traditionally, the efficiency of an overlay network is estimated with two key metrics which are the \emph{routing path length} between a pair of nodes and the \emph{maximum degree} of each node in the network.
State-of-the-art overlay networks achieve both a logarithmic routing path length and a logarithmic maximum degree~\cite{chord,FraigniaudG06}.

Up until recently, most work on overlay networks focused on demand-oblivious \emph{worst-case} scenarios, where no prior information about the possible future demand is known.
However, noting that most networking traffics are \emph{sparse} and partly \emph{predictable}~\cite{AvinGGS20}, a new line of work studies \emph{demand-aware} networks~\cite{SplayNet16,avin2017demand,SpiderDAN} which can adapt their connectivity properties to better fit the current traffic.\\
On another front, there has been a recent surge in \emph{learning-augmented} systems which, compared to traditional ones, take an additional parameter (the prediction) to improve their performances, with most recent examples in data structures~\cite{MLSkipList1,MLSkipList2}.
There, the goal is to use a \emph{predicted demand} for each node and adapt the system to better fit the incoming \emph{real demand}.
Since predicted and real demands may mismatch, they analyze what happens when the prediction is perfect (consistency), adversarial (robustness) and anything in between (smoothness, see~\cite{LykourisV18}).

This work considers an overlay network with $n$ nodes and assume it faces a certain real demand.
Like previous works on learning-augmented systems, our goal is to use a predicted demand to answer the overall demand with delays (cost) smaller than classic, prediction-oblivious systems.
We use the framework of consistency and robustness to assess the performances of our overlay algorithm; for that we compare the performances of our learning-augmented algorithm with those of a Skip List Network (\sln, \autoref{def:sln}) whose heights are optimal with respect to the real demand.\\
Unlike previous works however, our framework shows two crucial differences compared to other learning-augmented systems: our algorithm must (1) obey the standard constraints of overlay networks (\autoref{par:network-demand}), and (2) optimize on a \emph{pairwise demand} representing the intensity of the exchanges between any two nodes (\autoref{par:network-alg}).

\subsection{Our Setup}

\paragraph*{Network and Demand.}\label{par:network-demand}
We model an overlay network by an undirected graph with $n$ nodes.
An edge of the graph means that the two end-nodes can communicate directly without the help of other nodes, ie, they hold each other's ids.
This graph is an abstract representation of the routing possibilities between the nodes and may be different from the underlying, physical network that carries the communications.\\
At the considered time, we model a \emph{demand} between the nodes with a \emph{demand matrix} $W \in M_{n \times n}(\mathbb{R})$, where $W(u, v)$ is the amount of communications from $u$ to $v$.

\paragraph*{Overlay Network Algorithms.}\label{par:network-alg}
An overlay network algorithm is an algorithm distributed over the nodes of the network that defines the connectivity between overlay nodes and how routing can be performed.
On top of the physical, static network, the overlay algorithm keeps track of local routing information stored inside each node.

\paragraph*{Prediction Model.}
The initialization procedure of our overlay algorithm (run right before the join procedure) takes a \emph{prediction} as an additional argument (see also~\cite{P2PWithAdvice}).
As a newcomer node joins the overlay network, the prediction influences its neighborhood in the network to (ideally) fit the incoming demand best.
In the case of this paper, the prediction is the (integer) height of the node in the considered skip list network.
Adaptation to a changing demand (and thus changing prediction) is performed by leaving and re-joining the network.


\paragraph*{Consistency, Robustness, and Smoothness.}
The predictions are \emph{untrusted} and should not be followed blindly by the nodes.
In this context, we evaluate the performances of our algorithm with the two most frequently studied metrics called \emph{consistency} and \emph{robustness}~\cite{LykourisV18} (\autoref{def:consistency-robustness}).
In more details, we derive overlay algorithms that provably improve their performance in case of globally accurate predictions (consistency), but still present good performances (poly-logarithmic in $n$ with high probability) in case of arbitrary predictions, even adversarial ones (robustness).
In our distributed context however, we observe that the consistency metric may be too weak and its assumption that all nodes have good predictions unrealistic.
We therefore consider a third metric, \emph{smoothness}, which indicates the performances along the spectrum from consistency to robustness; in particular, what happens when the prediction is good but with a slight noise.
We provide extensive analysis of the smoothness of our algorithm tested on synthetic and real demands and with different types of errors (\autoref{sec:experiments}).

\subsection{Contributions}

The contributions of this paper are threefold.
First, we formulate the Optimum Static Skip List Network (\ossln) problem and prove it can be solved in polynomial time, with an algorithm based on dynamic programming (\autoref{OSSLN:optimum-cost}).
This result extends the previous work of~\cite{martinez1995optimal} to the case with pairwise demands in line with recent demand-aware results~\cite{SplayNet16}.
This algorithm demonstrates that the task of finding optimal predictions for our learning-augmented overlay algorithm is polynomial, which allowed us to run our evaluations.

As our second contribution, we present a new overlay design called \emph{Continuous Skip List Network} (\csln, \autoref{def:continuous_sln}), which generalizes Skip List Networks~\cite{SkipListNetwork} (\sln, \autoref{def:sln}) and the classic skip lists~\cite{skiplists}.
We show that a \csln offers a functioning overlay algorithm in that we implement the three basic primitives join, leave and route.
We then present a specific (without prediction) \csln called \uniform that draws its heights uniformly at random in the interval $[0, 1]$.
We show that \uniform matches the performances of state-of-the-art overlay algorithms with both routing path lengths and maximum degrees in $O(\log n)$ with high probability.

As our third contribution, we present a learning-augmented overlay algorithm, \laslin (Learning-Augmented Skip List Network), which combines our \uniform algorithm with a predicted \sln to strike consistency and robustness guarantees (\autoref{def:consistency-robustness}).
For each node, \laslin takes as prediction an integer between $1$ and $\lceil\ln n\rceil$ and sets it as its \emph{integer heights} at initialization time; it then draws a number in $[0, 1]$ uniformly at random and adds it to the integer heights to obtain its final heights as a \csln; it may then join the network using the join() procedure.
We show that \laslin is comparable to an optimum static \sln when the predictions are correct ($O(1)$-consistent) and at worst a $O(\log^2 n)$ factor away from the optimum static \sln in case of adversarial predictions ($O(\log^2n)$-robust).

\begin{figure*}[t]
\centering
\begin{minipage}{.45\textwidth}
  \centering
  \vspace{-1cm}
    \begin{tikzpicture}[scale=0.8]
      \pgfmathsetmacro{\gap}{0.98}
      \pgfmathsetmacro{\width}{0.5}
      \pgfmathsetmacro{\vsp}{0.01}
      \pgfmathsetmacro{\hsp}{0.03}
      \pgfmathsetmacro{\u}{3}
      \pgfmathsetmacro{\v}{10}
      \pgfmathsetmacro{\verticalgrad}{0.9}
      \pgfmathsetseed{44} 

      \pgfmathsetmacro{\htwo}{\verticalgrad * 1}
      \pgfmathsetmacro{\hthree}{\verticalgrad * 3}
      \pgfmathsetmacro{\hfour}{\verticalgrad * 3}
      \pgfmathsetmacro{\hfive}{\verticalgrad * 1}
      \pgfmathsetmacro{\hsix}{\verticalgrad * 3}
      \pgfmathsetmacro{\hseven}{\verticalgrad * 2}
      \pgfmathsetmacro{\height}{\verticalgrad * 3}
      \pgfmathsetmacro{\hnine}{\verticalgrad * 2}

      \draw[very thick, fill=ForestGreen!90] (\gap*2,0) -- ++(0,\htwo) -- ++(\width,0) -- ++(0,-\htwo);
      \draw[very thick, fill=ForestGreen!90] (\gap*3,0) -- ++(0,\hthree) -- ++(\width,0) -- ++(0,-\hthree);
      \draw[very thick, fill=ForestGreen!90] (\gap*4,0) -- ++(0,\hfour) -- ++(\width,0) -- ++(0,-\hfour);
      \draw[very thick, fill=ForestGreen!90] (\gap*5,0) -- ++(0,\hfive) -- ++(\width,0) -- ++(0,-\hfive);
      \draw[very thick, fill=ForestGreen!90] (\gap*6,0) -- ++(0,\hsix) -- ++(\width,0) -- ++(0,-\hsix);
      \draw[very thick, fill=ForestGreen!90] (\gap*7,0) -- ++(0,\hseven) -- ++(\width,0) -- ++(0,-\hseven);
      \draw[very thick, fill=ForestGreen!90] (\gap*8,0) -- ++(0,\height) -- ++(\width,0) -- ++(0,-\height);
      \draw[very thick, fill=ForestGreen!90] (\gap*9,0) -- ++(0,\hnine) -- ++(\width,0) -- ++(0,-\hnine);


      
      \draw[<->, thick] (2*\gap + \width + \hsp, \htwo - \vsp) -- (3*\gap - \hsp, \htwo - \vsp);

      \draw[<->, thick] (3*\gap + \width + \hsp, \hfour - \vsp) -- (4*\gap - \hsp, \hfour - \vsp);

      \draw[<->, thick] (4*\gap + \width + \hsp, \hfive - \vsp) -- (5*\gap - \hsp, \hfive - \vsp);
      \draw[<->, thick] (4*\gap + \width + \hsp, \hfour - \vsp) -- (6*\gap - \hsp, \hfour - \vsp);

      \draw[<->, thick] (5*\gap + \width + \hsp, \hfive - \vsp) -- (6*\gap - \hsp, \hfive - \vsp);
      
      \draw[<->, thick] (6*\gap + \width + \hsp, \hseven - \vsp) -- (7*\gap - \hsp, \hseven - \vsp);
      \draw[<->, thick] (6*\gap + \width + \hsp, \height - \vsp) -- (8*\gap - \hsp, \height - \vsp);

      \draw[<->, thick] (7*\gap + \width + \hsp, \hseven - \vsp) -- (8*\gap - \hsp, \hseven - \vsp);

      \draw[<->, thick] (8*\gap + \width + \hsp, \hnine - \vsp) -- (9*\gap - \hsp, \hnine - \vsp);

      \coordinate (P0) at (\gap*\u + 0.5*\gap,\hthree);
      \coordinate (P1) at (\gap*\u + 1*\gap,\hthree);
      \coordinate (P2) at (\gap*4,\hfour);
      \coordinate (P3) at (\gap*9,\hfour);
      \coordinate (P4) at (\gap*9,\hnine);
      \coordinate (P5) at (\gap*10,\hnine);

      \draw[very thick] (2*\gap,0) -- (9*\gap + \width,0);

      \draw[thick] (\gap*1.75,0) -- (\gap*1.75,\verticalgrad*4cm);
      \foreach \y in {0, 1, 2, 3, 4} {
        \draw[thick] (\gap*1.75 - 0.1, \verticalgrad*\y) -- (\gap*1.75 + 0.1, \verticalgrad*\y);
        \node at (\gap*1.75 - 0.3, \verticalgrad*\y) {\y};
      }

      \pgfmathsetmacro{\labelgap}{-0.3}
      \node at (\gap*2 + 0.5*\width, \labelgap) {1};
      \node at (\gap*3 + 0.5*\width, \labelgap) {2};
      \node at (\gap*4 + 0.5*\width, \labelgap) {3};
      \node at (\gap*5 + 0.5*\width, \labelgap) {4};
      \node at (\gap*6 + 0.5*\width, \labelgap) {5};
      \node at (\gap*7 + 0.5*\width, \labelgap) {6};
      \node at (\gap*8 + 0.5*\width, \labelgap) {7};
      \node at (\gap*9 + 0.5*\width, \labelgap) {8};

    \end{tikzpicture}
    \caption{
      A \sln with eight nodes.
      The node ids range from $1$ to $8$.
      The double-arrowed, horizontal lines are the edges of the \sln.
    }
    \label{fig:skip-list-network}   
  \end{minipage}
  \hspace{1cm}
\begin{minipage}{.45\textwidth}
  \vspace{-0.95cm}
  \centering
    \begin{tikzpicture}[scale=0.8]
      \pgfmathsetmacro{\gap}{0.98}
      \pgfmathsetmacro{\width}{0.5}
      \pgfmathsetmacro{\vsp}{0.01}
      \pgfmathsetmacro{\hsp}{0.03}
      \pgfmathsetmacro{\u}{3}
      \pgfmathsetmacro{\v}{10}
      \pgfmathsetmacro{\verticalgrad}{0.9}

      \pgfmathsetseed{38} 

      \pgfmathsetmacro{\htwo}{\verticalgrad*0.83}
      \pgfmathsetmacro{\hthree}{\verticalgrad*0.73}
      \pgfmathsetmacro{\hfour}{\verticalgrad*0.2}
      \pgfmathsetmacro{\hfive}{\verticalgrad*0.31}
      \pgfmathsetmacro{\hsix}{\verticalgrad*0.58}
      \pgfmathsetmacro{\hseven}{\verticalgrad*0.22}
      \pgfmathsetmacro{\height}{\verticalgrad*0.92}
      \pgfmathsetmacro{\hnine}{\verticalgrad*0.33}
      \pgfmathsetmacro{\hten}{\verticalgrad*0.34}

      \draw[very thick, fill=Orange!90] (\gap*2,1*\verticalgrad) -- ++(0,\htwo) -- ++(\width,0) -- ++(0,-\htwo);
      \draw[very thick, fill=Orange!90] (\gap*3,3*\verticalgrad) -- ++(0,\hthree) -- ++(\width,0) -- ++(0,-\hthree);
      \draw[very thick, fill=Orange!90] (\gap*4,3*\verticalgrad) -- ++(0,\hfour) -- ++(\width,0) -- ++(0,-\hfour);
      \draw[very thick, fill=Orange!90] (\gap*5,1*\verticalgrad) -- ++(0,\hfive) -- ++(\width,0) -- ++(0,-\hfive);
      \draw[very thick, fill=Orange!90] (\gap*6,3*\verticalgrad) -- ++(0,\hsix) -- ++(\width,0) -- ++(0,-\hsix);
      \draw[very thick, fill=Orange!90] (\gap*7,2*\verticalgrad) -- ++(0,\hseven) -- ++(\width,0) -- ++(0,-\hseven);
      \draw[very thick, fill=Orange!90] (\gap*8,3*\verticalgrad) -- ++(0,\height) -- ++(\width,0) -- ++(0,-\height);
      \draw[very thick, fill=Orange!90] (\gap*9,2*\verticalgrad) -- ++(0,\hnine) -- ++(\width,0) -- ++(0,-\hnine);

      \draw[<->, thick] (2*\gap + \width + \hsp, \htwo - \vsp+1*\verticalgrad) -- (3*\gap - \hsp, \htwo - \vsp+1*\verticalgrad);

      \draw[<->, thick] (3*\gap + \width + \hsp, \hfour - \vsp+3*\verticalgrad) -- (4*\gap - \hsp, \hfour - \vsp+3*\verticalgrad);
      \draw[<->, thick] (3*\gap + \width + \hsp, \hsix - \vsp+3*\verticalgrad) -- (6*\gap - \hsp, \hsix - \vsp+3*\verticalgrad);
      \draw[<->, thick] (3*\gap + \width + \hsp, \hthree - \vsp+3*\verticalgrad) -- (8*\gap - \hsp, \hthree - \vsp+3*\verticalgrad);

      \draw[<->, thick] (4*\gap + \width + \hsp, \hfive - \vsp+1*\verticalgrad) -- (5*\gap - \hsp, \hfive - \vsp+1*\verticalgrad);
      \draw[<->, thick] (4*\gap + \width + \hsp, \hfour - \vsp+3*\verticalgrad) -- (6*\gap - \hsp, \hfour - \vsp+3*\verticalgrad);

      \draw[<->, thick] (5*\gap + \width + \hsp, \hfive - \vsp+1*\verticalgrad) -- (6*\gap - \hsp, \hfive - \vsp+1*\verticalgrad);

      \draw[<->, thick] (6*\gap + \width + \hsp, \hseven - \vsp+2*\verticalgrad) -- (7*\gap - \hsp, \hseven - \vsp+2*\verticalgrad);
      \draw[<->, thick] (6*\gap + \width + \hsp, \hsix - \vsp+3*\verticalgrad) -- (8*\gap - \hsp, \hsix - \vsp+3*\verticalgrad);

      \draw[<->, thick] (7*\gap + \width + \hsp, \hseven - \vsp+2*\verticalgrad) -- (8*\gap - \hsp, \hseven - \vsp+2*\verticalgrad);

      \draw[<->, thick] (8*\gap + \width + \hsp, \hnine - \vsp+2*\verticalgrad) -- (9*\gap - \hsp, \hnine - \vsp+2*\verticalgrad);



      \pgfmathsetmacro{\hone}{1*\verticalgrad}
      \pgfmathsetmacro{\htwo}{1*\verticalgrad}
      \pgfmathsetmacro{\hthree}{3*\verticalgrad}
      \pgfmathsetmacro{\hfour}{3*\verticalgrad}
      \pgfmathsetmacro{\hfive}{1*\verticalgrad}
      \pgfmathsetmacro{\hsix}{3*\verticalgrad}
      \pgfmathsetmacro{\hseven}{2*\verticalgrad}
      \pgfmathsetmacro{\height}{3*\verticalgrad}
      \pgfmathsetmacro{\hnine}{2*\verticalgrad}
      \pgfmathsetmacro{\hten}{2*\verticalgrad}

      \draw[very thick, fill=ForestGreen!90] (\gap*2,0) -- ++(0,\htwo) -- ++(\width,0) -- ++(0,-\htwo);
      \draw[very thick, fill=ForestGreen!90] (\gap*3,0) -- ++(0,\hthree) -- ++(\width,0) -- ++(0,-\hthree);
      \draw[very thick, fill=ForestGreen!90] (\gap*4,0) -- ++(0,\hfour) -- ++(\width,0) -- ++(0,-\hfour);
      \draw[very thick, fill=ForestGreen!90] (\gap*5,0) -- ++(0,\hfive) -- ++(\width,0) -- ++(0,-\hfive);
      \draw[very thick, fill=ForestGreen!90] (\gap*6,0) -- ++(0,\hsix) -- ++(\width,0) -- ++(0,-\hsix);
      \draw[very thick, fill=ForestGreen!90] (\gap*7,0) -- ++(0,\hseven) -- ++(\width,0) -- ++(0,-\hseven);
      \draw[very thick, fill=ForestGreen!90] (\gap*8,0) -- ++(0,\height) -- ++(\width,0) -- ++(0,-\height);
      \draw[very thick, fill=ForestGreen!90] (\gap*9,0) -- ++(0,\hnine) -- ++(\width,0) -- ++(0,-\hnine);

      \coordinate (P0) at (\gap*\u + 0.5*\gap,\hthree);
      \coordinate (P1) at (\gap*\u + 1*\gap,\hthree);
      \coordinate (P2) at (\gap*4,\hfour);
      \coordinate (P3) at (\gap*9,\hfour);
      \coordinate (P4) at (\gap*9,\hnine);
      \coordinate (P5) at (\gap*10,\hnine);
      \coordinate (P6) at (\gap*10,\hten);
      \coordinate (P7) at (\gap*10 + 0.5*\gap,\hten);

      \draw[very thick] (2*\gap,0) -- (9*\gap + \width,0);

      \draw[very thick] (2*\gap,0) -- (9*\gap + \width,0);

      \draw[thick] (\gap*1.75,0) -- (\gap*1.75,\verticalgrad*4cm);
      \foreach \y in {0, 1, 2, 3, 4} {
        \draw[thick] (\gap*1.75 - 0.1, \verticalgrad*\y) -- (\gap*1.75 + 0.1, \verticalgrad*\y);
        \node at (\gap*1.75 - 0.3, \verticalgrad*\y) {\y};
      }

      \pgfmathsetmacro{\labelgap}{-0.3}
      \node at (\gap*2 + 0.5*\width, \labelgap) {1};
      \node at (\gap*3 + 0.5*\width, \labelgap) {2};
      \node at (\gap*4 + 0.5*\width, \labelgap) {3};
      \node at (\gap*5 + 0.5*\width, \labelgap) {4};
      \node at (\gap*6 + 0.5*\width, \labelgap) {5};
      \node at (\gap*7 + 0.5*\width, \labelgap) {6};
      \node at (\gap*8 + 0.5*\width, \labelgap) {7};
      \node at (\gap*9 + 0.5*\width, \labelgap) {8};

    \end{tikzpicture}
    \caption{
      A \csln obtained by adding each height of \autoref{fig:skip-list-network} with a random number in $[0, 1]$.
      This is a possible outcome of \laslin with predicted heights in green.
    }
    \label{fig:laslin}
  \end{minipage}
\end{figure*}

\section{Model and Metrics}

In this section, we formally define our model and the metrics we use to evaluate our overlay algorithms.

\begin{definition}[\sln]
\label{def:sln}
A Skip List Network (\sln) $\mathcal{N}$ is a triple $(V, h, E)$ where:
  \begin{itemize}
  \item
    $V$ is a totally ordered set of node ids;
  \item
    $h$ is a height function $ h : V \rightarrow \mathbb{N}^+$;
  \item
    $E$ is the set of edges such that $\{u, v\} \in E$ if and only if
    \begin{align*}
       \forall x \in V \; \text{ s.t. }\; u < x < v, \quad h_x < \min(h_u, h_v)
    \end{align*}
  \end{itemize}
\end{definition}


\noindent
We redefine hereafter the greedy, local routing path specific to \sln.
Intuitively, the routing path goes to the next neighbor that is the closest to the destination (in terms of ids) without ``flying over'' the destination.
\begin{definition}[Routing path]
 \label{def:routing-path}
 Let $\mathcal{N} = (V, h, E)$ be a \sln and $u, v \in V$ be two different nodes.
 The routing path from $u$ to $v$ in $\mathcal{N}$ is the sequence of edges $e_1, e_2 \dots e_k$ such that (1) $u \in e_1$, (2)$v \in e_k$, and (3) $\forall i \in [k]$, if $u < v$ then $e_i = \{x, y\}$ where $y$ is the largest node neighbor of $x$ such that $x < y \le v$, otherwise if $v < u$ then $e_i = \{x, y\}$ where $y$ is the smallest node neighbor of $x$ such that $v \le y < x$
\end{definition}

\begin{definition}[Cost]
  \label{def:cost}
  Let $\mathcal{N} = (V, h, E)$ be an \sln and $W$ be a demand matrix.
  The cost function is:
  \begin{align*}
    \text{Cost}(\mathcal{N}, W) = \sum_{u, v \in V} W(u, v) \cdot d_{\mathcal{N}}(u, v),
  \end{align*}
  where $d_{\mathcal{N}}(u, v)$ is the routing path length between $u$ and $v$ in $\mathcal{N}$.
\end{definition}

\begin{definition}[Optimum Static Skip List Network (OSSLN) problem]
  \label{def:ossln}
  Given a set of nodes $V$ and a demand matrix $W$, the objective of the Optimum Static \sln (OSSLN) problem is to find a height assignment $h: V\rightarrow \mathbb{N}^+$ that minimizes the cost of the  \sln $\mathcal{N}=(V, h, E)$.
\end{definition}

\begin{definition}[Overlay Algorithm]
  \label{def:overlay-algorithm}
  An overlay algorithm is a distributed algorithm running on each node of the network, which provides three procedures:
  \begin{itemize}
  \item
    \textbf{initialization:} Executed before a node joins the network.
    In this paper, sets the height of each node in the \sln. 
  \item
    \textbf{join:} Executed to join the network.
    Takes as argument the id of a node that is already in the network.
  \item
    \textbf{leave:} Executed to leave the network.
  \end{itemize}
\end{definition}

\begin{definition}[Learning-Augmented Overlay Algorithm]
  \label{def:learning-augmented}
  A learning-augmented overlay algorithm ALG is an overlay algorithm whose initialization procedure takes an additional parameter, the input prediction.
  In this paper, the input prediction consists of an estimated height assignment for each node. 
\end{definition}

\begin{definition}[Consistency and Robustness]
  \label{def:consistency-robustness}
  We say that a learning-augmented \csln algorithm is
  \begin{itemize}
  \item
    $\alpha$-consistent if, for all demand matrices $W$ and all set of nodes $V$, there exists a set of predictions (assignment of heights) such that the expected cost of the algorithm is at most $\alpha$ times the cost of an optimal optimum static \sln for $V$ and $W$.
  \item
    $\beta$-robust if, for any demand matrix $W$, any set of nodes $V$, and any set of prediction (assignment of heights), the expected cost of the algorithm is at most $\beta$ times the cost of an optimal optimum static \sln for  $V$ and $W$. 
  \end{itemize}
\end{definition}

\section{Optimum Static Skip List Network}
\label{sec:ossln}

Given a set of nodes $V$ and a demand matrix $W$, the objective of the Optimum Static Skip List Network (OSSLN) problem is to find a height assignment $h: V\rightarrow \mathbb{N}^+$ that minimizes the total communication cost (\autoref{def:cost}) of the  \sln $\mathcal{N}=(V, h, E)$ (\autoref{def:sln}):
\begin{equation}
    \ossln(V, W) = \min_h\;Cost(\mathcal{N}, W). \label{eq:OSSLN-objective-function}
\end{equation}
where $\mathcal{N}$ is an \sln with node set $V$ and heights $h$.\\

In \cite{SplayNet16}, the authors presented a dynamic programming algorithm to find the Optimum static distributed Binary Search Tree (OBST) in polynomial time ($O(n^3)$). This algorithm leverages a crucial property: once a node $v$ is chosen as the root of a tree induced by the nodes in a given interval, the subproblems induced by the node id intervals to the left and to the right of $v$ can be solved independently.
The \ossln problem exhibits a similar structure, but with a crucial difference. 

If we define a \textit{pivot node} $x\in [l,r]$, $1\leq l\leq r\leq n,$ such that $h_x = \max_{w\in [l,r]}h_w$, $x$ also partitions the \ossln problem into two subproblems (on the intervals $[l,v)$ and $(v,r]$). However, a critical distinction arises when $h_x < \max_{w\in [l,r]}h_w$, in which case $x$ will not be a part of all the paths between node pairs $(u, v)$ $u\in [l, x), v\in (x, r]$. Specifically, for a node $x\in [l+1,r-1]$ such that $h_{l}, h_{r} > h_x$, any path between nodes $(u, v)$ where $u \leq l$ and $v \geq r$ does not include $x$.
This changes the optimal substructure of the problem, requiring a more sophisticated dynamic programming approach.

Given a \sln $\mathcal{N}$, we define a \textit{path} $P_h(l, r)$ to be the set of nodes in the \textit{routing path} between the nodes $l$ and $r$ (\autoref{def:routing-path}), excluding $l$, i.e., $l\notin P_h(l,r)$. This is a crucial simplification that allows the problem to be decomposed. Under this assumption, we have $d(l, r) = |P_h(l, r)|$. W.l.o.g., we assume that $P_h(l,r) = P_h(r, l)$.

Next, we define $L(x, h) = \min\{i \leq x\;|\;h_j \leq h_x,\;\forall j \in [i, x]\}$, i.e., the node of minimum id between the set of all (consecutive) nodes of height lower or equal to $h_x$, to the \textit{left} of $x$, and, analogously, $R(x, h) = \max\{i \geq x\;|\;h_j \leq h_x,\;\forall j \in [x, i]\}$, i.e., the node with maximum id between the set of all nodes of height lower or equal to $h_x$, to the \textit{right} of $x$. A node $x \in V$ belongs to a path $P_h(u,v)$ if and only if there does not exist any pair of nodes $(u', v')$ such that $u \leq u' < x$, $x < v' \leq v$, and (both) $h_{u'}, h_{v'} > h_x$. 

Finally, given an arbitrary $x \in V$, we define $P_x(h)$ as the set of all pairs $(u, v) \in V \times V$, such that $x \in P_h(u,v)$. 

\begin{lemma}\label{lemma:OSSLN-path-belong}
If $h$ is a solution for the $OSSLN(V,W)$, then $\forall x\in V$ we have that:
\begin{eqnarray}
    P_x(h) &=& 
    \{ (u, v) \mid (u < x) \wedge v \in [x, R(x, h)] \}\cup \nonumber\\
    && \{ (u, v) \mid u \in [L(x, h), x)] \wedge (v > R(x, h)) \}.\nonumber
     \label{eq:set-of-paths-of-x}
\end{eqnarray} 
\end{lemma}

\begin{proof}
If $u < x$ and $x\leq v \leq R(x,h)$, then, by definition of $R(x,h)$, $h_x \geq h_v \geq \min(h_u, h_v)$. Since $x$ lies between $u$ and $v$ and there are no edges ``above" $h_x$  in the interval $[u, v]$ (i.e., $\not\exists (w,z)\in E$, s.t. $h_w > h_x$ and $h_z > h_x$, $[w,z]\subseteq [u,v]$), $x$ must be a part of $P_h(u, v)$.

Similarly, if $L(x,h) \leq u < x$ and $v > R(x, h)$, then, again by definition of $L(x, h)$, $h_x \geq h_u \geq \min(h_u, h_v)$. Since $x$ lies between $u$ and $v$ and there are no edges "above" $h_x$ in the interval $[u, v]$, $x$ must be a part of $P_h(u, v)$

It remains to show that no other pair $(u,v)$ belongs to $P_x(h)$. If $u \geq x$, then by definition of $P_h(u,v)$, $x \notin P(u, v)$. Otherwise, suppose $u < L(x, h)$ and $v > R(x, h)$. Let $u' = L(x, h) - 1$ and $v' = R(x, h) - 1$, by definitions of $L(x,h)$ and $R(x,h)$, we have that $h_x < h_{u'}$ and $h_x < h_{v'}$. Since $u \leq u' < x < v' \leq v$, we have that:
\begin{equation}
    |P_h(u, u') \cup \{v'\} \cup P_h(v', v)| < |P_h(u, x) \cup \{x\} \cup P_h(x, v)|.\nonumber
\end{equation}
Because $P_h(u, v)$ is the shortest path from $u$ to $v$, it follows that $x \notin P_h(u, v)$. \qedhere
\end{proof}

Since $d(u, v) = |P_h(u, v)|$, we can now rewrite the total cost of an \sln $\mathcal{N}=(V,h,E)$ as:
\begin{eqnarray}
  Cost(\mathcal{N}, W) &=& \sum_{u, v \in V} W(u, v) \cdot d(u, v) \nonumber\\
  &=& \sum_{u, v \in V} \sum_{x \in P_h(u, v)} W(u, v) \nonumber\\
  &=& \sum_{x \in V} \sum_{(u, v) \in P_x(h)} W(u, v). \nonumber
\end{eqnarray}

This means that the total cost is equal to the sum of contributions of all nodes $x \in V$, where the \textit{contribution of a node} $x$ to the total cost is defined as:
\begin{equation}
\sum_{(u, v) \in P_x(h)} W(u, v) \label{eq:vertex-cost-contribution}
\end{equation}

\noindent\textbf{OSSLN with Unbounded height (OSSLNU):}
We now present a dynamic programming algorithm for the \ossln (\autoref{def:ossln}). Furthermore, we show that the algorithm's complexity can be reduced to $\mathcal{O}(n^3)$.

A crucial observation is that, for any solution to the OSSLNU problem, there exists a solution of cost less than or equal to that of the former, in which, $\forall [l,r], 1\leq l\leq r\leq n$, there is a unique \emph{pivot} node, i.e., a unique node whose height is the maximum height in that interval.

\begin{lemma} \label{lemma:OSSLNu-unique-height}
Let $h$ be an optimum solution to $\text{OSSLNU}(V, W)$. Then, there must exist a solution $h'$ such that $Cost((V,h',E'), W) \leq Cost((V, h, E), W)$, where, for any interval $[l, r], 1\leq l\leq r\leq n,$ there is exactly one node $x \in [l, r]$ with height $h'_x = \max_{w\in [l,r]}h'_w$.
\end{lemma}

\begin{proof}
Suppose, for contradiction, that any possible optimal solution contains at least one interval $[l, r]$ with at least two vertices
$x, y \in [l, r]$ with $h_x = h_y = \max(h_{l, r})$ and $x \neq y$.

Define $h'$ as follows for any vertex $v \in V$:

$
h'_v = 
    \begin{cases} 
    h_v + 1 & \text{if } v = x \text{ or } h_v > h_x \\
    h_v & \text{otherwise} 
    \end{cases}
$

Note that, while $h'_x > h'_y$, this height assignment is order preserving. If $h_a > h_b$ then $h'_a > h'_b$.

If $x \notin P_h(u, v)$, since we preserved the order between $x$ and all vertices taller than $x$, $x$ cannot appear as an obstacle between any such vertices. Analogously, since we preserved the order between all vertices with at most the same height as $x$ there can't be any new edges between vertices with at most the same height as $x$, so the path is unmodified. 

If $x \in P_h(u, v)$, we know that:
\[
    P_h(u,v) = P_h(u, x)\;\cup \{x\} \; \cup P_h(x, v).
\]
From our definition of path, we know that $x \notin P_h(x,v)$, and therefore $P_h(x, v) = P_{h'}(x, v)$. Now, let's consider a vertex $z$, such that $u \leq z < x$ and $h_z > h_x$. If such a vertex $z$ does not exist, then we can infer that $P_h = P_{h'}$. Furthermore, if there does exist such a vertex $z$, then:
\[
    |P_{h'}| = |P_h(u, z)\;\cup \{x\}\; \cup P_h(x, v)| \leq |P_h|.
\]
Since $\text{Cost}((V,h,E), W) = \sum_{(u, v) \in V} W(u, v) \cdot d(u, v)$ and $d(u, v) = |P_h(u, v)|$, and any path in $h'$ is of at most the same size as in $h$, then, $\text{Cost}((V,h',E'), W) \leq \text{Cost}((V,h,E), W)$. Hence, the lemma holds.
\end{proof}

In OSSLNU, because we are only interested in the relative order of node heights, rather than in their absolute values, we have that $(nl = l)$ and $(nr = r)$, $\forall [l, r], 1\leq l\leq r\leq n$, and the height parameter $h$ becomes unnecessary. This simplifies the recurrence relation, reducing the complexity to $\mathcal{O}(n^3)$.

\begin{theorem}\label{OSSLNu:optimum-cost}
Consider an interval $[l, r]$, $1\leq l \leq r \leq n$.
Let $\text{OSSLNU}(l, r, W)$ denote the minimum total cost contributed by all nodes $v \in [l, r]$. 
Let $x \in [l, r]$ be a candidate pivot node. Define $l$ as the leftmost index such that $\forall x, \in[l,r],\forall v \in [l, x],\, h_v \leq h_x$, and similarly define $r$ as the rightmost index such that $\forall v \in [x, r],\, h_v \leq h_x$. Then, the recurrence relation $\text{OSSNU}(l, r, W)=$
\begin{eqnarray}
    \min_{x \in [l, r]} \Big(
        \text{OSSNU}(l, x - 1, W) + 
        \text{OSSNU}(x + 1, r, W) + \nonumber\\
        \sum_{u < x} \sum_{v\in[x, r]} W(u, v)
        + \sum_{u \in [l, x)}\sum_{v > r}W(u, v)\Big). \nonumber 
\end{eqnarray}
\end{theorem}
\begin{proof}
The proof is by induction on the size of the interval $[l, r]$.

Note that, by Lemma 2 once a pivot vertex $x\in[l,r]$ is chosen as the highest vertex within the interval $[l, r]$, we can assume, without loss of generality, that all vertices in the optimal solutions to the left subinterval $[l, x - 1]$ and to the right subinterval $[x + 1, r]$ have heights strictly lower than $h_x$. 

Therefore, for a solution $h$ to $\text{OSSLNU}(l, r, W)$, it follows that $l = L(x, h)$ and $r = R(x, h)$.
    
\noindent\textbf{Base Case $(l = r)$:} In this case, the interval contains a single vertex $x = l = r$.  In this case, the cost reduces to the contribution of $x$ to the solution cost, which, from Lemma 1, and $L(x, h) = l$ and $R(x, h) = r$ as we demonstrated above, it follows that this expression corresponds exactly to:
\begin{eqnarray}
  \sum_{u < x} \sum_{v\in[x, r]} W(u, v) 
  + \sum_{u \in [l, x)}\sum_{v > r}W(u, v) \label{eq:vertex-cost-contribution}
\end{eqnarray}
\noindent\textbf{Inductive Step:} Suppose the recurrence holds for all subintervals of $[l, r]$ of size less than $(r - l + 1)$. For an interval $[l, r]$ of size greater than one, we consider a candidate pivot vertex $x \in [l, r]$. 

The total cost of placing $x$ as pivot of the interval $[l, r]$ equals to the sum of:
\begin{itemize}
    \item The optimal cost for the left subinterval $[l, x-1]$,
    \item The optimal cost for the right subinterval $[x+1, r]$,
    \item The cost of all paths that include $x$, which by Lemma 1 is given by the Equation 4.
\end{itemize}

Since we minimize this total over all pivot choices $x \in [l, r]$, the recurrence holds.
\end{proof}

\noindent\textbf{OSSLN with Bounded Height:} While OSSLNu optimizes the total routing cost, it offers no guarantees regarding individual node degrees or the average network degree. Consequently, a single node could theoretically accumulate a near-linear number of edges, rendering a physical network implementation infeasible.

Since the degree of a node is upper-bounded by twice its height, we can control the maximum degree by enforcing a height limit, $h_{\max}$. For the discrete SLN, capping the height at $h_{\max}$ strictly ensures that the maximum node degree does not exceed $2 \times h_{\max}$. 

\begin{definition}[\ossln with Bounded Height $h_{\max}$]
  \label{def:ossln-bounded}
  An \ossln with bounded height $h_{\max}$ $\mathcal{N_H} = (V, h, E)$ is the discrete \sln (\autoref{def:sln}) that minimizes the communication cost given by \autoref{eq:OSSLN-objective-function} such that the height assignment for any of its nodes is at most $h_{\max}$.
\end{definition}

In the \ossln with bounded height, the height assignment is of the form $h: V\longrightarrow [1,h_{\max}]$, where $h_{\max}$ is a parameter, typically bounded by $O(\log{n})$ or $O(1)$.

Consider an interval $[l, r]$, $1\leq l \leq r \leq n$ and a height bound $h\in[1,h_{\max}]$.
Let $nl$, $1\leq nl \leq l$, denote the leftmost node id, such that $\forall v \in [nl, l], h_v \leq h$, and, similarly, let $nr$, $r\leq nr \leq n$, denote the rightmost node id, such that $\forall v \in [r, nr],\, h_v \leq h$.
Finally, let $\ossln(l, r, nl, nr, h, W)$ denote the minimum total cost contributed by all nodes $v \in [l, r]$. 

\begin{theorem}\label{OSSLN:optimum-cost}
If $x \in [l, r]$ is a candidate pivot node with height $h\in [1,h_{\max}]$, then the following recurrence relation holds $\forall (nl,l, x, r,nr)$, $1\leq nl\leq l\leq x\leq r\leq nr\leq n$:
{ \footnotesize
\begin{eqnarray}
\text{OSSLN}(l, r, nl, nr, h, W) = 
\min \Bigg\{\text{OSSLN}(l, r, l, r, h - 1, W), \nonumber\\
    \min_{x \in [l, r]} \Bigg[
    \text{OSSLN}(l, x - 1, nl, nr, h, W) +\nonumber\\
    \text{OSSLN}(x + 1, r, nl, nr, h, W) + \nonumber\\
    \sum_{u < x} \sum_{v \in [x, nr]} W(u, v) + \sum_{u \in [nl, x)} \sum_{v > nr} W(u, v) \Bigg] \Bigg\}.\nonumber
\end{eqnarray}
}
\end{theorem}
\begin{proof}
The proof is by strong induction on the size of the interval $[l,r]\subseteq [1,n]$:

Note that, by definition, $nl = L(x, h)$ and $nr = R(x, h)$ for any  possible $h$. Furthermore, by construction, for an interval induced by $l$, $r$, $nl$, $nr$ with bounded height $h$, assigning a pivot vertex with height equal to $h-1$, we have $nl' = l$ and $nr' = r$.

\noindent\textbf{Base case $(l = r)$:} In this case, we have that $(x = l)$  and the set of paths $x$ belongs to is given by Lemma 1 By setting $nl=L(x, h)$ and $nr=R(x, h)$ in the proof, the vertex contribution to the total cost is given by:
\begin{equation}
  \sum_{u < x} \sum_{v\in[x, nr]} W(u, v) 
  + \sum_{u \in [nl, x)}\sum_{v > nr}W(u, v) \label{eq:ossg-pivot-contribution}
\end{equation}

\noindent\textbf{Induction step:} For $(h = 1)$, we can disregard the first half of the recursion, where we assign the pivot vertex a height lower than $h$. The total contribution of all vertices to the cost becomes:
\begin{equation}
  \sum_{x\in[l, r]} \left( \sum_{u < x} \sum_{v\in[x, nr]} W(u, v) 
  + \sum_{u \in [nl, x)}\sum_{v > nr}W(u, v) \right), \label{eq:ossg-pivot-contribution} \nonumber
\end{equation}
\noindent which, from the base case, is precisely what the second half of the recurrence equation will return for any choice of the pivot vertex.

Let $1\leq l < r\leq n$ and $h\in[1,h_{\max}]$. Assume that the recurrence equation is true for all $h', l', r'$ such that $(h' < h)~\wedge~(r'-l') < (r - l)$. From previous step also assume that the equation is true for $(r'-l') = (r - l)$ with height $(h' < h)$.

Since, when assigning the pivot vertex with height lower than $h$, 
by construction, we have that $(nl = l)$ and $(nr = r$). The minimum cost for the solution for any $nl$ and $nr$ is the minimum between keeping the height of the pivot, $(h_x=h)$, or decreasing it, $h_x\leq(h-1)$, 
which is given by the minimum between:
\begin{eqnarray}
    \min_{x \in [l, r]} \Bigg[
        \ossln(l, x - 1, nl, nr, h, W) + \nonumber\\
     \ossln(x + 1, r, nl, nr, h, W) + \nonumber\\
    \sum_{u < x} \sum_{v \in [x, nr]} W(u, v) + \sum_{u \in [nl, x)} \sum_{v > nr} W(u, v)
    \Bigg] \nonumber
\end{eqnarray}
and $OSSLN(l, r, l, r, h-1, W)$, which is precisely what we want to prove.
\end{proof}

The complexity of the above dynamic programming algorithm is $\mathcal{O}(n^5\cdot h_{\max})$, if we pre-compute an auxiliary matrix of aggregate weights $W'(x,y)=\sum_{z\leq y}W(x,z), \forall (x,y), 1\leq x\leq y\leq n$. 

\section{Continuous Skip-List Networks} \label{sec:continuous-sln}

This section introduces a generalization of \sln called \csln.

\begin{definition}[Continuous Skip List Network (\csln)]
  \label{def:continuous_sln}
  A \csln $\mathcal{N} = (V, h, E)$ is an \sln (\autoref{def:sln}) except the height function $h$ is defined over the real numbers, i.e., $h: V \rightarrow \mathbb{R}$.
\end{definition}

In this section, we show that \csln offer a set of properties similar to classic \sln with arguably simpler proofs.
Next, we detail key procedures of \csln, namely \textbf{join}, \textbf{leave} and \textbf{route}.

\subsection{Procedures}
\SetAlgorithmName{Procedure}{Procedure}{List of Procedures}

We first present the internal variables used in our procedures. In particular, a node $u$ holds the variables:
\begin{itemize}
  \item $u$: the unique id of the node
  \item $h_u$: the height in $\mathbb{R}$ of node $u$
  \item $N_u$: the set of all neighbors of $u$
  \item $D_u$: dictionary mapping a neighbor $w \in N_u$ to its height
  \item $\text{anchor}_u$: the node with an id closest to $u$, either the successor or predecessor of $u$.
\end{itemize}

\paragraph*{Route.}
The route algorithm (\autoref{alg:route}) takes the destination node id $v$ and hops from node to node until it finds the destination node, returning $(v, True)$ or reaches a dead end in some node $w \neq v$, return $(w, False)$.

\begin{algorithm}[t]
\caption{route($v$)}\label{alg:route}
\SetKwInOut{Input}{Input}
\tcp{code running at node $u$}
\If{$u$ = $v$}{
  \Return ($u$, True)
}
\ElseIf{$u < v$}{
  $S \gets \{x \in N_u | u < x \leq v\}$\\
  \If{$S = \emptyset$}{
    \Return ($u$, False)
  }
  \Else{
    $w \gets \max S$ \\
    \Return $w.\text{route}(v)$
  }
}
\Else{
  $S \gets \{x \in N_u | v \le x < u\}$\\
  \If{$S = \emptyset$}{
    \Return ($u$, False)
  }
  \Else{
    $w \gets \min S $\\
    \Return $w.\text{route}(v)$
  }
}
\end{algorithm}

\paragraph*{Join.}
The join algorithm (\autoref{alg:join}) inserts the node inside the Continuous \sln based on its anchor local variable.
The potential neighbor changes of the other nodes propagate through the network with the find\_neighbors algorithm (\autoref{alg:find_neighbors}, in this paper's appendix).

\begin{algorithm}[t]

\caption{join()}\label{alg:join}
\tcp{code running at node $u$}
\If{$\text{anchor} < u$}
{
  $S \gets \{x \in N_{\text{anchor}_u} | \text{anchor} < x\}$\\
}
\Else
{
  $S \gets \{x \in N_{\text{anchor}_u} | x < \text{anchor}\}$\\
}
\If{$S = \emptyset$}{
  $N_u \gets N_u \cup \text{anchor}_u.\text{find\_neighbors}(u, h_v, -\infty)$\\
  Update $D_u$
}
\Else{
  \If{$\text{anchor} < u$}{
    $\text{succ} \gets \min S$\\
  }
  \Else{
    $\text{succ} \gets \max S$\\
  }
  $N_u \gets N_u \cup \text{anchor}_u.\text{find\_neighbors}(u, h_v, -\infty)$\\
  $N_u \gets N_u \cup \text{succ}.\text{find\_neighbors}(u, h_v, -\infty)$\\
  Update $D_u$
}

\end{algorithm}

\paragraph*{Leave.}
The leave algorithm (\autoref{alg:leave}, in this paper's appendix) shows how an already-inserted node can leave the network.
The potential neighbor changes propagate through the network with the delete\_neighbors algorithm (\autoref{alg:delete_neighbors}, in this paper's appendix).

\begin{algorithm}[t]

\caption{leave()}\label{alg:leave}
\tcp{code running at node $u$}
\If{$\text{anchor} < u$}
{
  $S \gets \{x \in N_{\text{anchor}_u} | \text{anchor} < x\}$\\
}
\Else
{
  $S \gets \{x \in N_{\text{anchor}_u} | x < \text{anchor}\}$\\
}
\If{$S = \emptyset$}{
  $N_u \gets N_u \cup \text{anchor}_u.\text{delete\_node}(u, h_v, -\infty)$\\
  Update $D_u$
}
\Else{
  \If{$\text{anchor} < u$}{
    $\text{succ} \gets \min S$\\
  }
  \Else{
    $\text{succ} \gets \max S$\\
  }
  $N_u \gets N_u \cup \text{anchor}_u.\text{delete\_node}(u, h_v, -\infty)$\\
  $N_u \gets N_u \cup \text{succ}.\text{delete\_node}(u, h_v, -\infty)$\\
  Update $D_u$
}

\end{algorithm}

\subsection{\uniform \csln}
We now present a simple \csln called \uniform whose heights are chosen uniformly at random in $[0, 1]$, and show that \uniform has the performances of state-of-the-art overlay network algorithms as both its routing search path and maximum degree are logarithmic with high probability.
The properties of \uniform will be used to derive our main result on \laslin in the next section.



\begin{protocol}[\uniform]
  \uniform is a \csln whose initialization procedure sets its heights to a real number in $[0, 1]$ chosen uniformly at random.
\end{protocol}


\begin{theorem}[Expected Routing Path Length]
  \label{thm:expected-path-length}
  For \uniform, the routing path between any two nodes has length smaller than $2 \ln n$ on expectation.
\end{theorem}
\begin{proof}
  Let $u$ and $v$ be two nodes.
  We note $X$ the random variable equal to the routing path length between $u$ and $v$.
  We note $X_{\text{up}}$ the number of height increases along the path, likewise $X_{\text{down}}$ is the number of path decreases.
  It holds that $X = X_{\text{up}} + X_{\text{down}}$.
  We will give an upper bound of $X_{\text{up}}$, the same upper bound will hold for $X_{\text{down}}$ by symmetry.
  Let $w$ be a node between $u$ and $v$ in the \sln, we note $1_{w}$ the indicator random variable equal to $1$ if (1) $w$ is part of the path between $u$ and $v$, and (2) if the height increases when hoping to $w$.
  It holds $X_{\text{up}} = \sum_{w \in [u+1, v]} 1_{w}$.\\
  We now compute $P[1_{w} = 1]$ for any $w \in [u+1, v]$.
  $w$ is part of the increasing path if and only if the height of $w$ is strictly greater than every other heights in $[u, w-1]$.
  The probability that two uniform random variables are equal is zero, hence there exists exactly one node with maximum height in $[u, w]$.
  The heights are independent and uniformly distributed, hence $w$ is the node with maximum height with probability $1 / (|w - u| + 1)$:
  \begin{align*}
    \forall w \in [u+1, v], \quad P[1_w = 1] = \frac{1}{|w - u| + 1}
  \end{align*}

  \noindent
  Noticing that $\mathbb{E}[1_w] = P[1_w = 1]$ and gathering the two previous equalities gives:
  \begin{align*}
    \mathbb{E}[X_{\text{up}}] &= \sum_{w \in [u+1, v]} \frac{1}{|w - u| + 1} = \sum_{i=2}^{|v-u|+1} \frac{1}{i} \le \ln n
  \end{align*}
  The last inequality holds as the harmonic sum until $n$, $1 + 1/2 + \dots + 1/n$, is lower than $1 + \ln n$.
  The claimed upper bound on the expected path length follows as the same reasoning can be done for $X_{\text{down}}$.
\end{proof}

\begin{theorem}[Routing Path Length with High Probability]
  \label{thm:path-length-whp}
  For \uniform, the routing path between any two nodes has length smaller than $2e \cdot \ln n$ with probability greater than $1 - \frac{2}{n}$.
\end{theorem}
\begin{proof}
  Let $u$ and $v$ be two nodes such that $u < v$, and let $X_{\text{up}}$ be the random variable equal to the increasing routing path length between $u$ and $v$ .
  Let $w \in [u, v]$, we note $1_{w}$ the indicator random variable equal to $1$ if $w$ is part of the increasing path between $u$ and $v$ (see proof of Theorem 3).

  We first show that all $1_w$ are independent random variables.
  Let $w$ and $w ^{\prime}$ be two nodes such that $u < w < w ^{\prime}$, it holds:
  \begin{align*}
    P[1_{w ^{\prime}} = 1 \mid 1_w = 1] &= \\ P[h_x < h_{w ^{\prime}} \forall x \in [w+1, w ^{\prime} - 1] \cap \max_{y \in [u, w]} h_y < h_w ^{\prime}]
                                        &= \\ P[h_y < h_w ^{\prime} \forall y \in [u, w ^{\prime} - 1]]
                                        &= \\ P[1_{w ^{\prime}} = 1]
  \end{align*}

  Now that we showed independence, we use the following Chernoff's bound: For $Y$ the sum of independent Bernoulli random variables (a Poisson trial) and $\mu = \mathbb{E}[Y]$, $\forall \delta > 0$ it holds that:
  \begin{align*}
    P(Y \ge (1+\delta) \mu) < \left(\frac{e^{\delta}}{(1+\delta)^{1+\delta}}\right)^{\mu}
  \end{align*}
  Replacing $Y = X_{\text{up}} = \sum_{w \in [u+1, v]} 1_w$, $\mu = \ln n$ and $\delta = e - 1$ gives:
  \begin{align*}
    P(X_{\text{up}} \ge e \cdot \ln n) &< \frac{e^{\delta \cdot \ln n}}{(1+\delta)^{(1+\delta) \cdot \ln n}}\\
                                       &= \frac{n^{\delta}}{n^{(1+\delta) \cdot \ln (1+\delta)}} = \frac{1}{n}
  \end{align*}
  By symmetry, the same result holds for $X_{\text{down}}$.
  Finally, we obtain the desired result with a union bound:
  \begin{align*}
    P(X \ge 2e \cdot \ln n) &\le P(X_{\text{up}} \ge e \cdot \ln n \cup X_{\text{down}} \ge e \cdot \ln n)\\
                            &\le P(X_{\text{up}} \ge e \cdot \ln n) + P(X_{\text{down}} \ge e \cdot \ln n)\\
                            &< \frac{2}{n}\qedhere
  \end{align*}
\end{proof}

\begin{theorem}[Expected Degree]
  \label{thm:expected-degree}
  For \uniform, the expected degree of any node is lower than $4$.
\end{theorem}

\begin{proof}
  If $u$ and $w$ are two nodes, we define $N_u$ the random variable equal to the number of neighbors of $u$, and $1_{u, w}$ the indicator random variable equal to one if $u$ and $w$ are neighbors.
  It holds $N_u = \sum_{w \neq u} 1_{u, w}$.
  Two nodes $u$ and $w$ with $u < w$ are neighbors if and only if $u$ and $w$ are the first and second highest nodes in $[u, w]$, hence it holds that $\mathbb{E}[1_{u, w}] = \frac{2}{|w-u| \cdot (|w - u| + 1)}$.
  It therefore holds:
  \begin{align*}
    \mathbb{E}[N_u] &=  \sum_{w \in [1, n] \setminus \{u\}} \frac{2}{|w-u| \cdot (|w - u| + 1)} 
                    \\ &\le  4 \cdot \sum_{i=1}^{n/2} \frac{1}{i \cdot (i+1)} = 4 \cdot \sum_{i=1}^{n/2} \frac{1}{i} - \frac{1}{i+1}\\
                    &= 4 \cdot (1 - \frac{1}{n/2+1}) \le 4\qedhere
  \end{align*}
\end{proof}

\begin{theorem}[Maximum Degree with High Probability]
  \label{thm:degree-whp}
  For \uniform, the maximum degree is lower than $4 \ln n$ with probability greater than $1 - \frac{1}{n^{3}}$.
\end{theorem}
\begin{proof}
  Let $u$, $w$ and $w ^{\prime}$ be three different nodes.
  We name $1_{u, w}$ and $1_{u, w ^{\prime}}$ the indicator random variables equal to $1$ if $u$ and $w$ are neighbors, and if $u$ and $w ^{\prime}$ are neighbors, respectively.
  We first show that $1_{u, w}$ and $1_{u, w ^{\prime}}$ are independent random variables: the proof is almost identical to the one of Theorem 4.

  We first establish that each node taken individually has a degree lower than $\ln n$ with probability greater than $1 - \frac{1}{n^4}$.
  Let $N_u$ be the random variable equal to the number of neighbors of $u$, we use the Chernoff bound for Poisson trials used in the proof of \autoref{thm:path-length-whp}:
  \begin{align*}
    P[N_u \ge  4 \ln n] &< \left(\frac{e^{(\ln n) - 1}}{(\ln n)^{\ln n}}\right)^{4} = \frac{(n / e)^{4}}{n^{4 \ln \ln n}} < \frac{1}{n^4}
  \end{align*}
  This last equality holds when $n$ is greater than a given constant, which is not a problem as we are interested in cases when $n$ is large.
  We conclude with a union bound: the probability that at least one node has a degree greater than $4 \ln n$ is upper bounded by $\sum_{v \in V} P[N_v \ge  4 \ln n] \le \frac{1}{n^3}$ and the claim holds.
\end{proof}

\section{Learning-Augmented \csln}
\label{sec:laslin}

In this section, we present \laslin, a \emph{learning-augmented} \csln algorithm, where each node can obtain a prediction of the future demand matrix at any time.

We first show that, for general demands, any demand-oblivious overlay algorithm can be off by a factor of $O(\log n / \log \log n)$ compared to an optimum static \sln.
This result provides a comparison point around which the consistency and robustness metrics of demand-aware algorithms can be evaluated.

\begin{theorem}
  \label{thm:lb_non_la}
  We consider any overlay algorithm that is not learning-augmented (\autoref{def:learning-augmented}).
  Then, there exists a demand matrix $W$ such that the considered algorithm has a cost of $\Omega(\log n / \log \log n) \cdot |W|$ in expectation while an optimal static \sln algorithm pays $|W|$, where $|W|$ is the total sum of the elements of $W$.
\end{theorem}
\begin{proof}
  As the considered algorithms has a logarithmic maximum degree with high probability, there exists a constant $c$ such that the generated network has a maximum degree lower than $O (\ln^c n)$ with probability greater than $1 - O(1/n)$.
  Let $u$ be a node.
  We first cover the case when the maximum degree is lower than $\ln^c n$, all the considered probabilities and expected values are therefore conditional.
  In this case, we consider the disk centered around $u$ with smallest radius such that it contains no less than half of the nodes.
  The maximum distance between $u$ and a node in that disk is no lower than $O(\log n / \log \log n)$ as the network has a polylogarithmic maximum degree --- this result can be easily obtained by upper bounding all the degrees by $O(\log^c n)$.
  There exists another node $v$ such that $v$ is at the border or outside of the disc with probability at least $1/2$.
  We consider the demand matrix that is zero everywhere except at the coefficient with row $u$ and column $v$.
  An optimal \sln algorithm can clearly pay a cost of $1$, the considered \sln therefore has an approximation ratio worse of $\omega(\log n / \log \log n)$ for that demand matrix.\\
  Finally, we cover the cases where the maximum degree is higher than $\ln^c n$: That case only contributes with an additive constant to the distance between $u$ and $v$ as a distance is lower than $n$.
  We considered all the cases and the claim follows.
\end{proof}

\begin{protocol}[\laslin]
  \laslin is a learning-augmented \csln that takes an integer in $[1, \lceil\ln n\rceil]$ as a prediction.
  Its initialization procedure set its height by adding the input prediction with a real number in $]0, 1[$ chosen uniformly at random.
\end{protocol}

\begin{theorem}
  \label{thm:log_square_robustness}
  \laslin is a \csln that is $O(1)$-consistent and $O(\log^2 n)$-robust, while its maximum degree is $O(\log^2 n)$ with high probability.
\end{theorem}
\begin{proof}
  We start by proving $O(1)$-consistency.
  Adding values in $]0, 1[$ to the node's heights does not remove the edges of the initial \sln.
  Let $u$ and $v$ be two nodes and consider the routing path $P_{u,v} = u, x_1, x_2 \dots v$ from $u$ to $v$ in the initial \sln.
  Now calling $P_{u,v}'$ the path from $u$ to $v$ in the newly created \csln, we have that $P_{u, v}' \subseteq P_{u, v}$; this can be show with a simple induction on the path length starting from $u$.
  All paths of the \csln have equal or shorted length than in the initial \sln with perfect predictions and $O(1)$-consistency follows.\\
  We now prove the robustness claim.
  A routing path between any two nodes $u$ and $v$ has two phases: first ascending (the heights always increase) and then descending (the heights always decrease).
  For both phase in the initial \sln, the number of different encountered integer heights is upper-bounded by $O(\log n)$.
  For each encountered integer height, the expected routing path length in the \csln is upper-bounded by $O(\log n)$ (Theorem~\ref{thm:expected-path-length}) with high probability, we deduce that the routing path is lower than $O(\log^2 n)$ with high probability using a union bound on the integer heights of the initial \sln.
  We now use this upper bound of each distance in the total weighted path lengths: $\sum_{u, v \in V} W(u, v) \cdot d(u, v) \le O(\log^2 n) \cdot |W|$ and the optimal static \sln pays at least $|W|$.

  Finally, we show the claim on the maximum degree.
  Let $u$ be a node.
  Fixing an integer $i \in [0, p_u]$, in the newly created \csln $u$ has fewer than $\ln n$ neighbors with initial integer heights $i$ with high probability (Theorem~\ref{thm:degree-whp}).
  We complete the proof with a union bound on all possible $O(\log n)$ heights of the initial \sln.
\end{proof}


\begin{figure}[t]
\centering	
\includegraphics[width=.8\columnwidth]{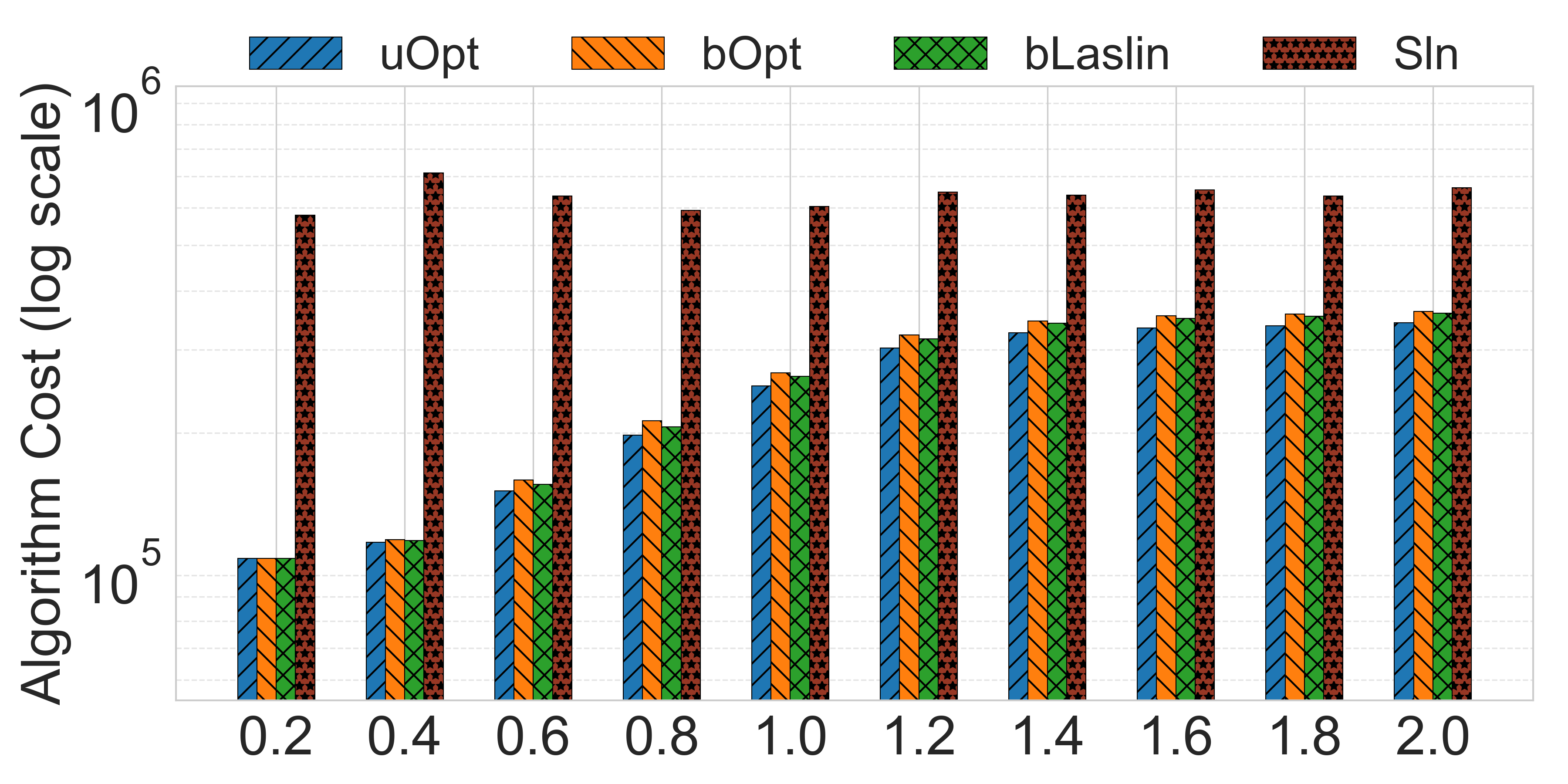}
\caption{\textbf{Consistency} (Pareto $\alpha\in[0.2, 2]$, $n=64$, $\eta=0$).}
\label{fig:result-pareto-consistency}
\end{figure}

\begin{figure*}[t]
    \centering
    \begin{subfigure}{1\linewidth}
        \centering
        \includegraphics[width=0.35\linewidth]{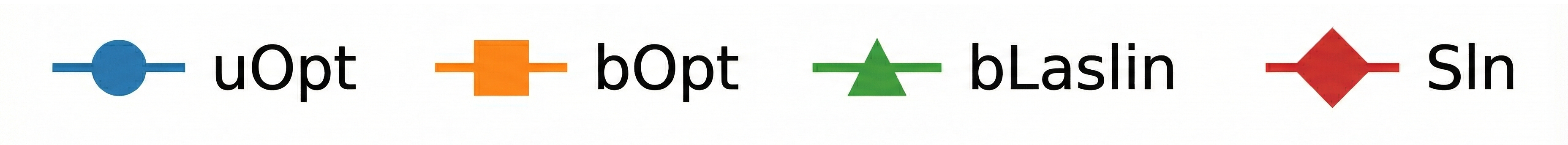}
        \vspace{-2mm} 
    \end{subfigure}
    
    \begin{subfigure}{0.19\linewidth}
        \centering
        \includegraphics[width=1\linewidth]{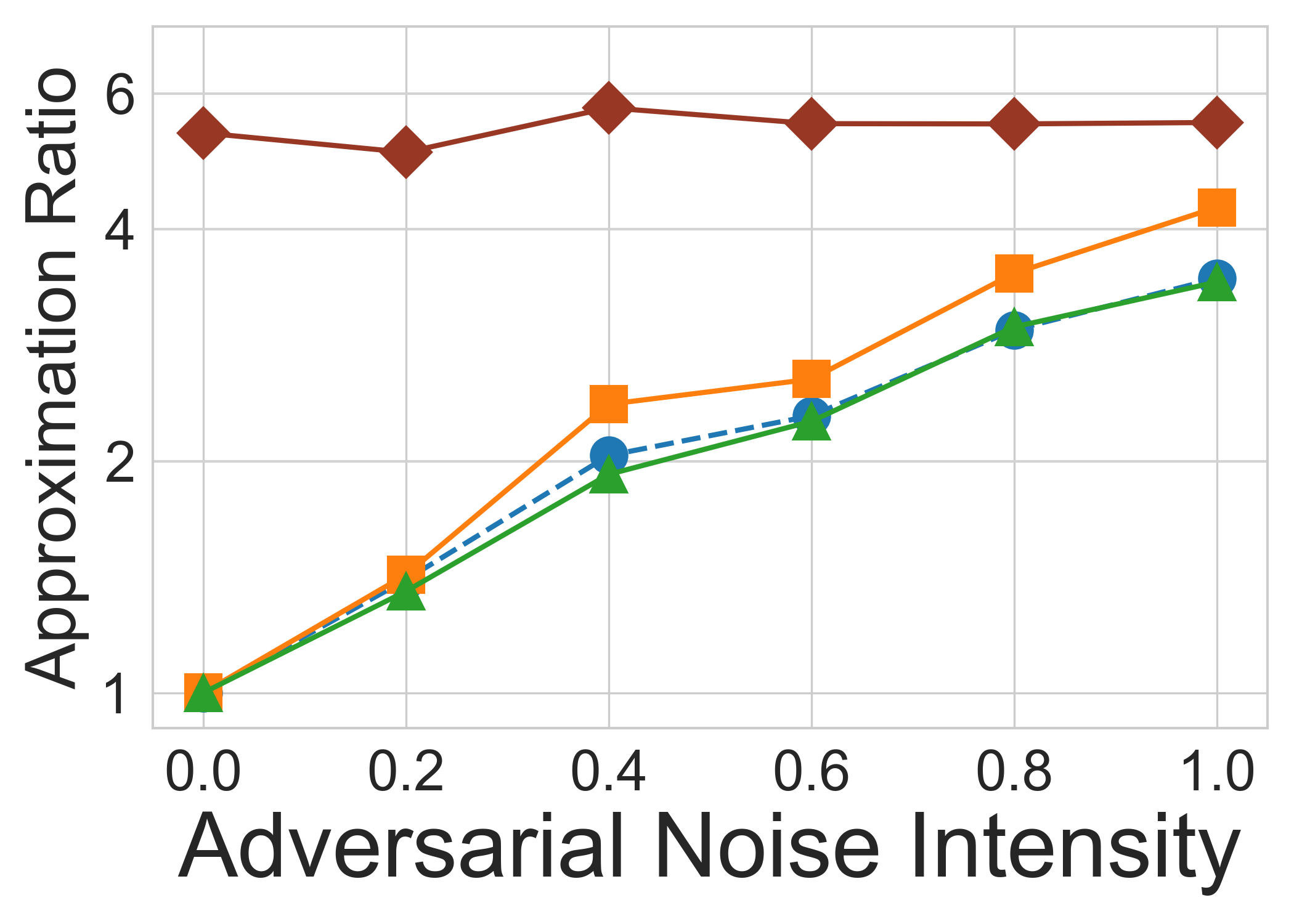}
        \caption{$\alpha = 0.2$}
    \end{subfigure}\hfill
    \begin{subfigure}{0.19\linewidth}
        \centering
        \includegraphics[width=1\linewidth]{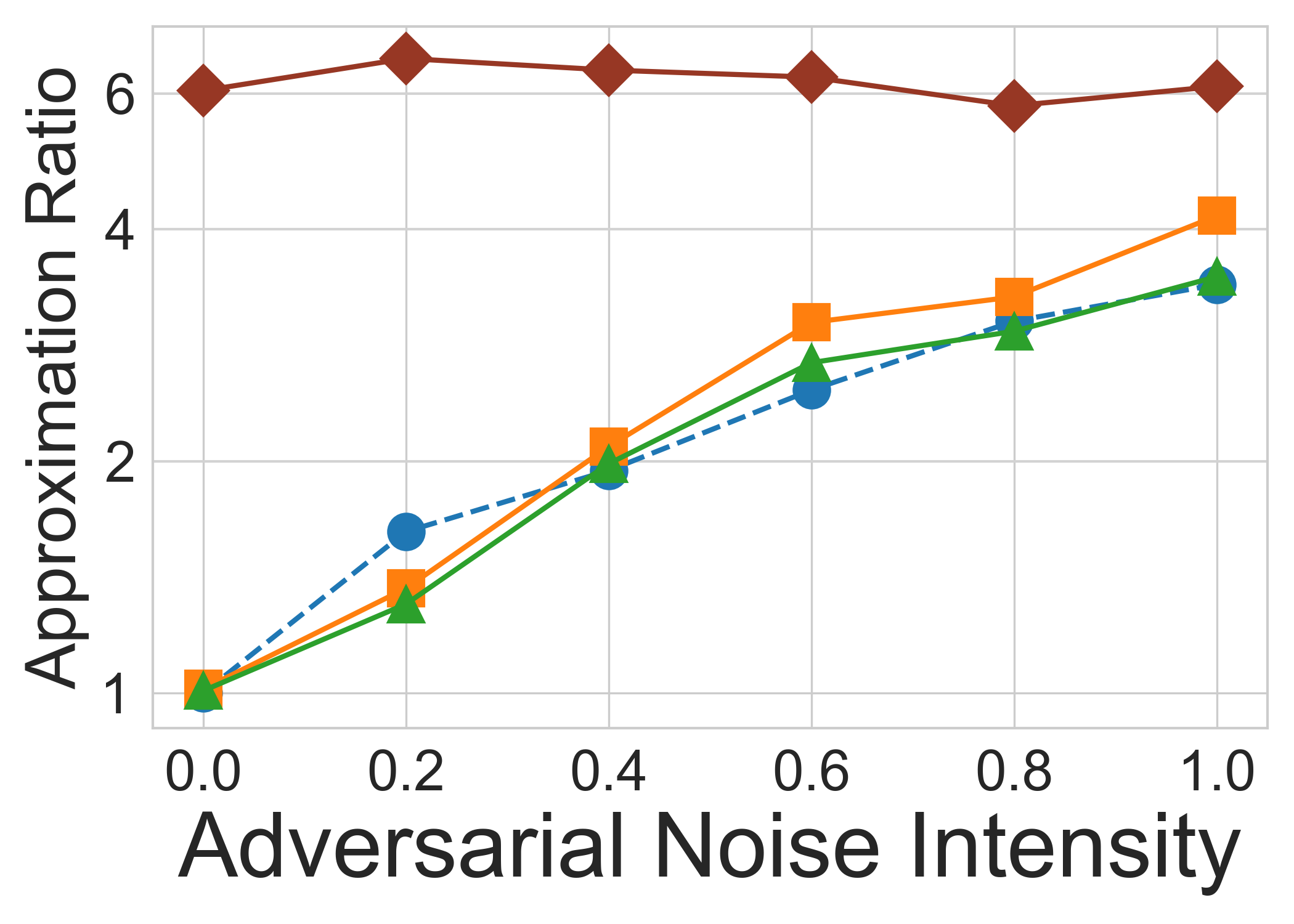}
        \caption{$\alpha = 0.4$}
    \end{subfigure}\hfill
    \begin{subfigure}{0.19\linewidth}
        \centering
        \includegraphics[width=1\linewidth]{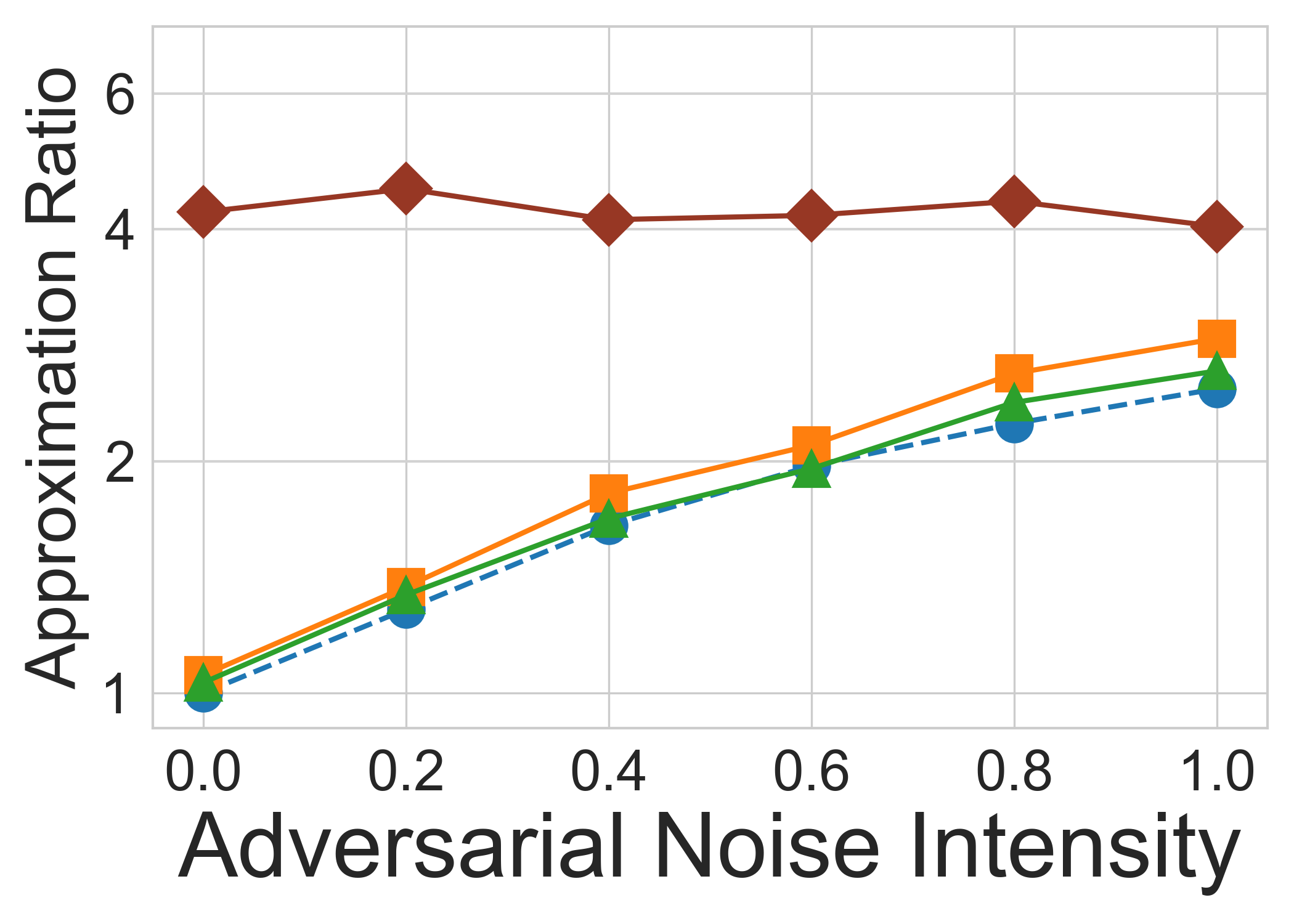}
        \caption{$\alpha = 0.6$}
    \end{subfigure}\hfill
    \begin{subfigure}{0.19\linewidth}
        \centering
        \includegraphics[width=1\linewidth]{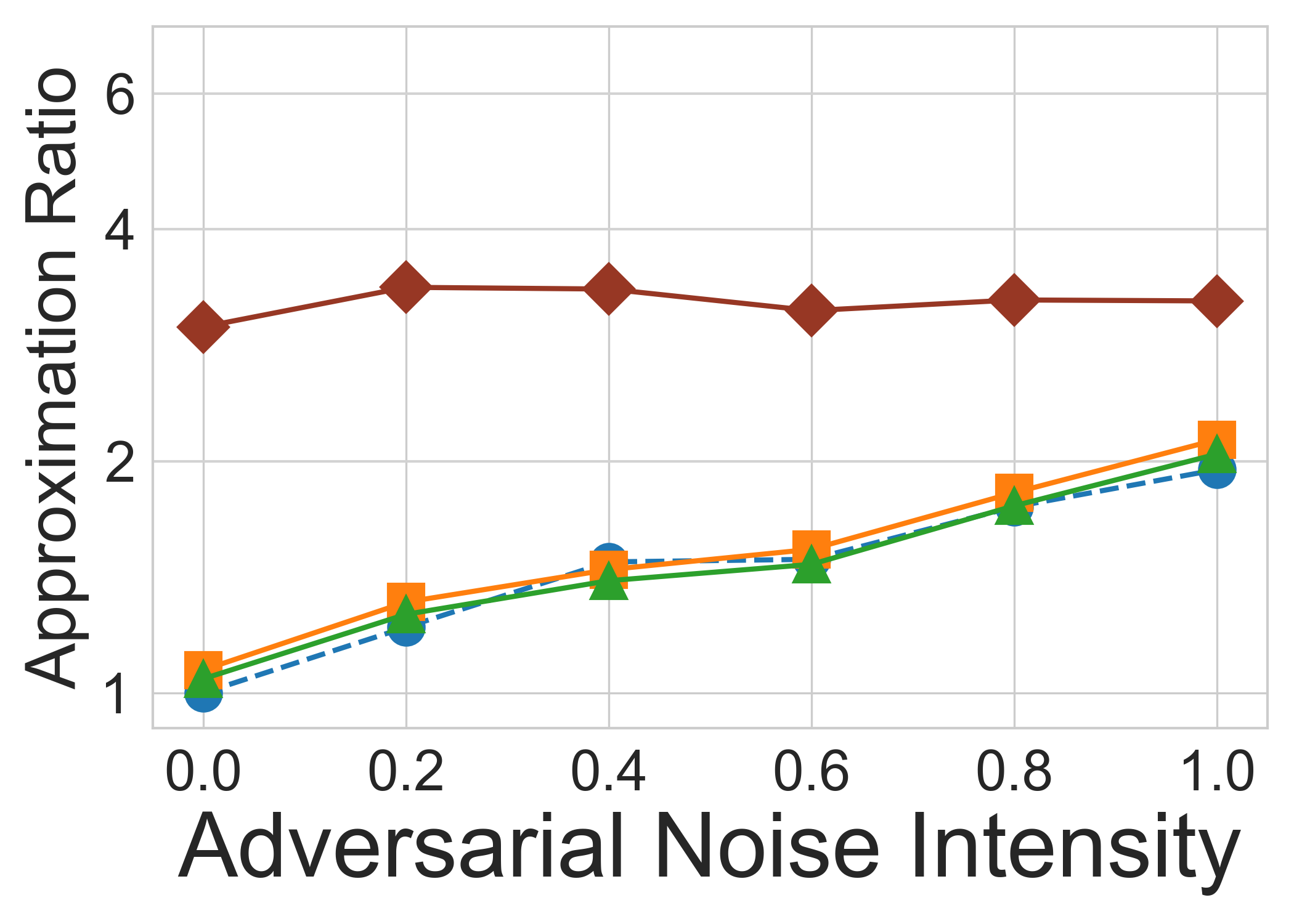}
        \caption{$\alpha = 0.8$}
    \end{subfigure}\hfill
    \begin{subfigure}{0.19\linewidth}
        \centering
        \includegraphics[width=1\linewidth]{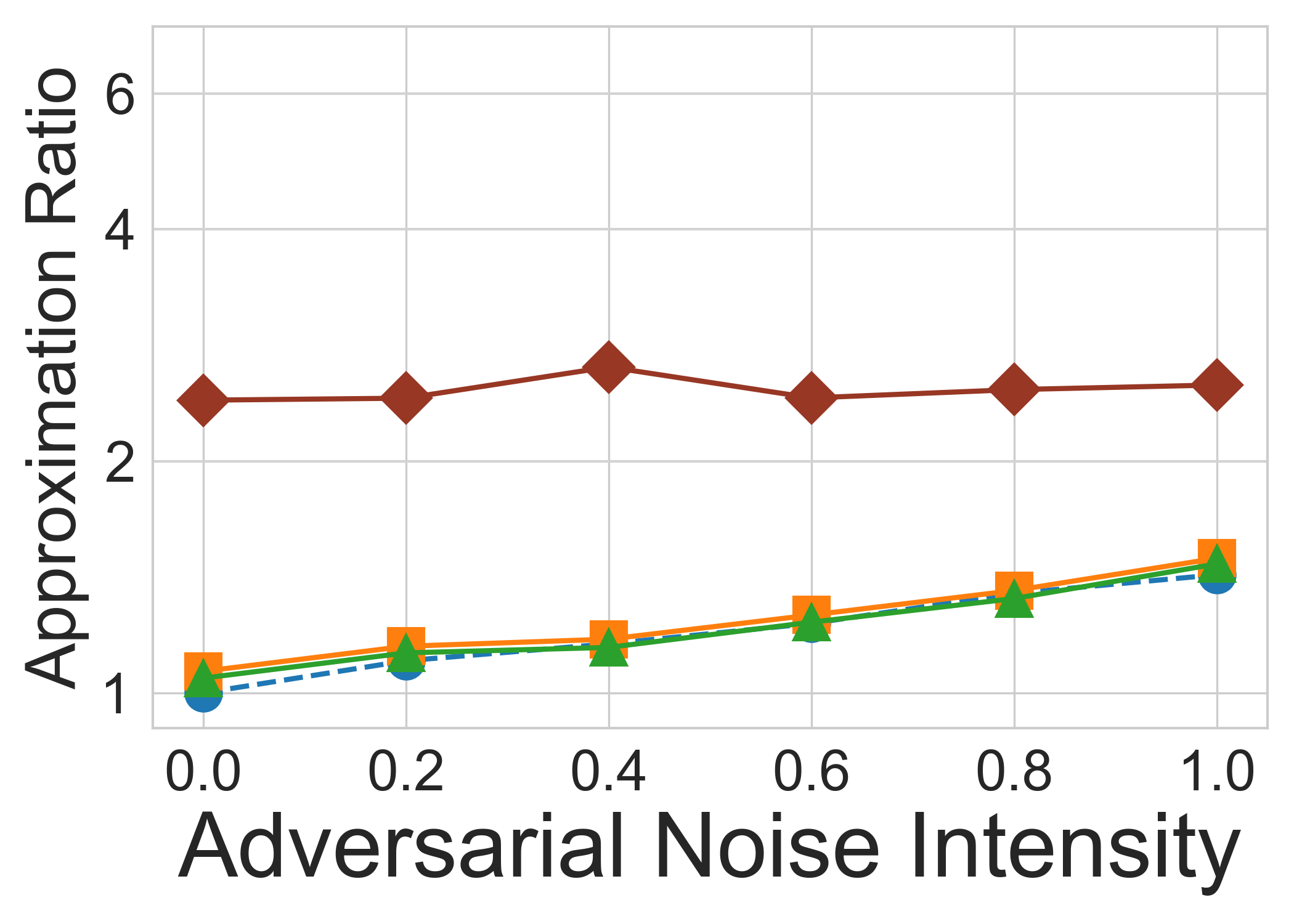}
        \caption{$\alpha = 1.0$}
    \end{subfigure}
    
    \vspace{2mm} 

    \begin{subfigure}{0.19\linewidth}
        \centering
        \includegraphics[width=1\linewidth]{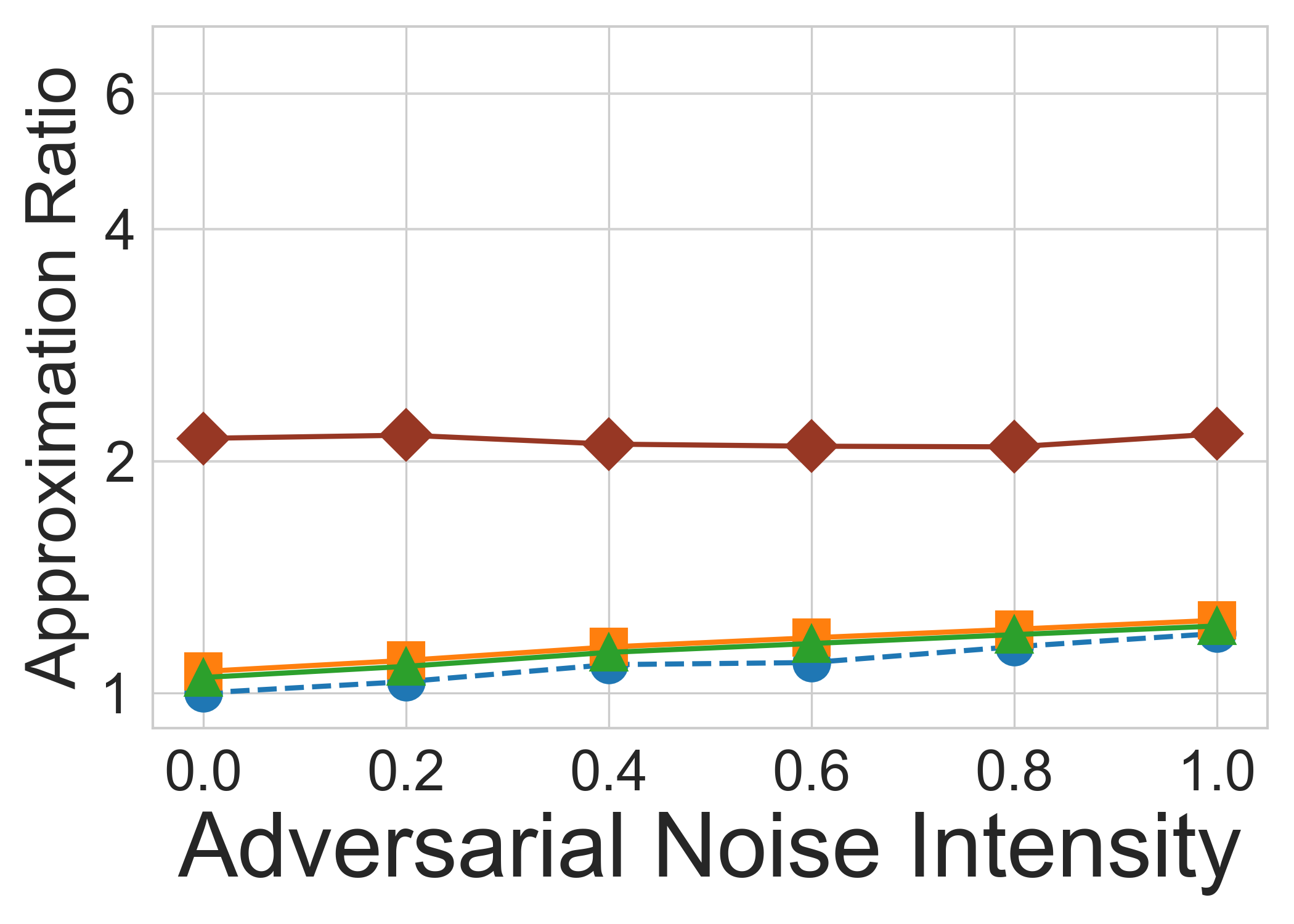}
        \caption{$\alpha = 1.2$}
    \end{subfigure}\hfill
    \begin{subfigure}{0.19\linewidth}
        \centering
        \includegraphics[width=1\linewidth]{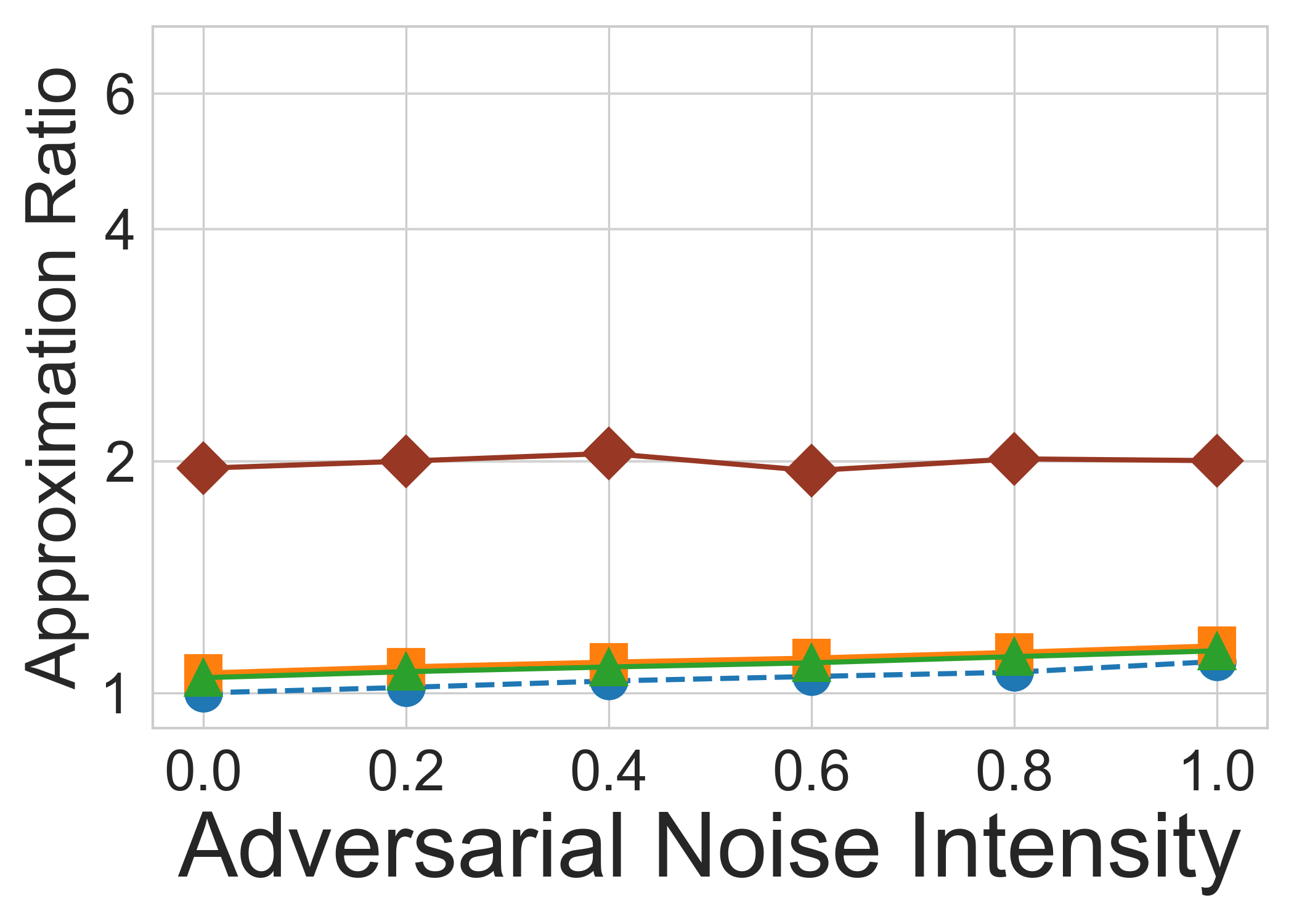}
        \caption{$\alpha = 1.4$}
    \end{subfigure}\hfill
    \begin{subfigure}{0.19\linewidth}
        \centering
        \includegraphics[width=1\linewidth]{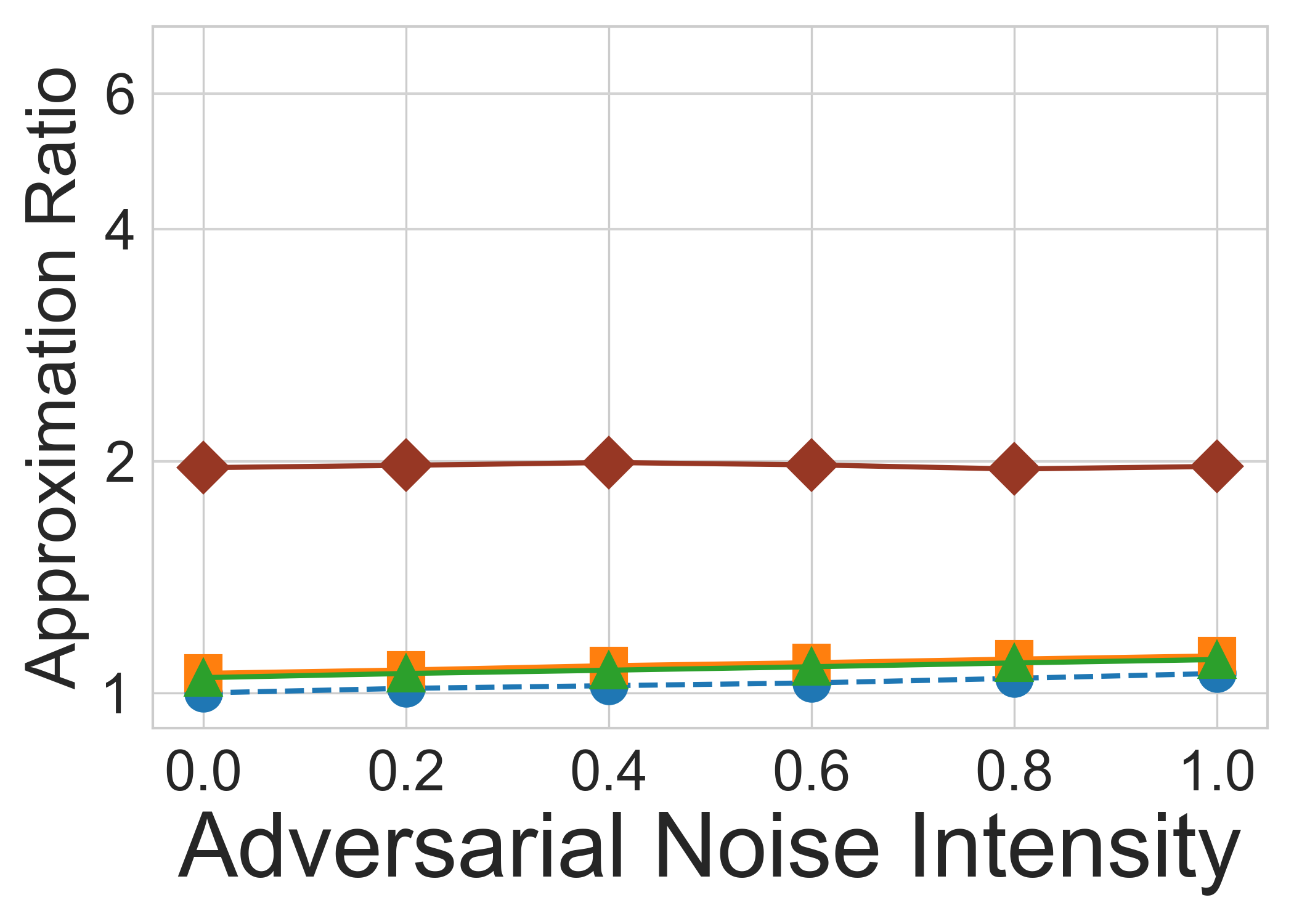} 
        \caption{$\alpha = 1.6$}
    \end{subfigure}\hfill
    \begin{subfigure}{0.19\linewidth}
        \centering
        \includegraphics[width=1\linewidth]{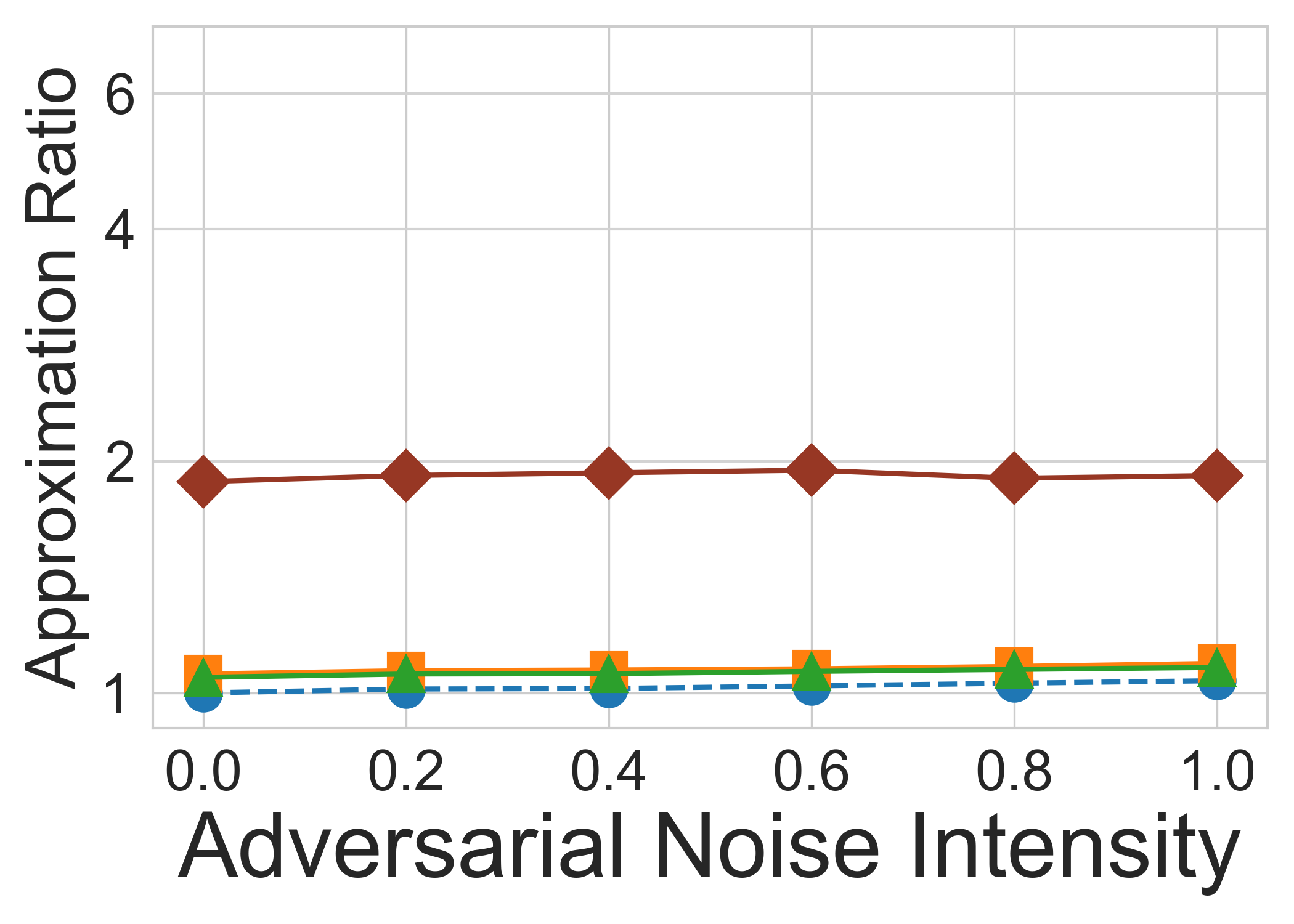}
        \caption{$\alpha = 1.8$}
    \end{subfigure}\hfill
    \begin{subfigure}{0.19\linewidth}
        \centering
        \includegraphics[width=1\linewidth]{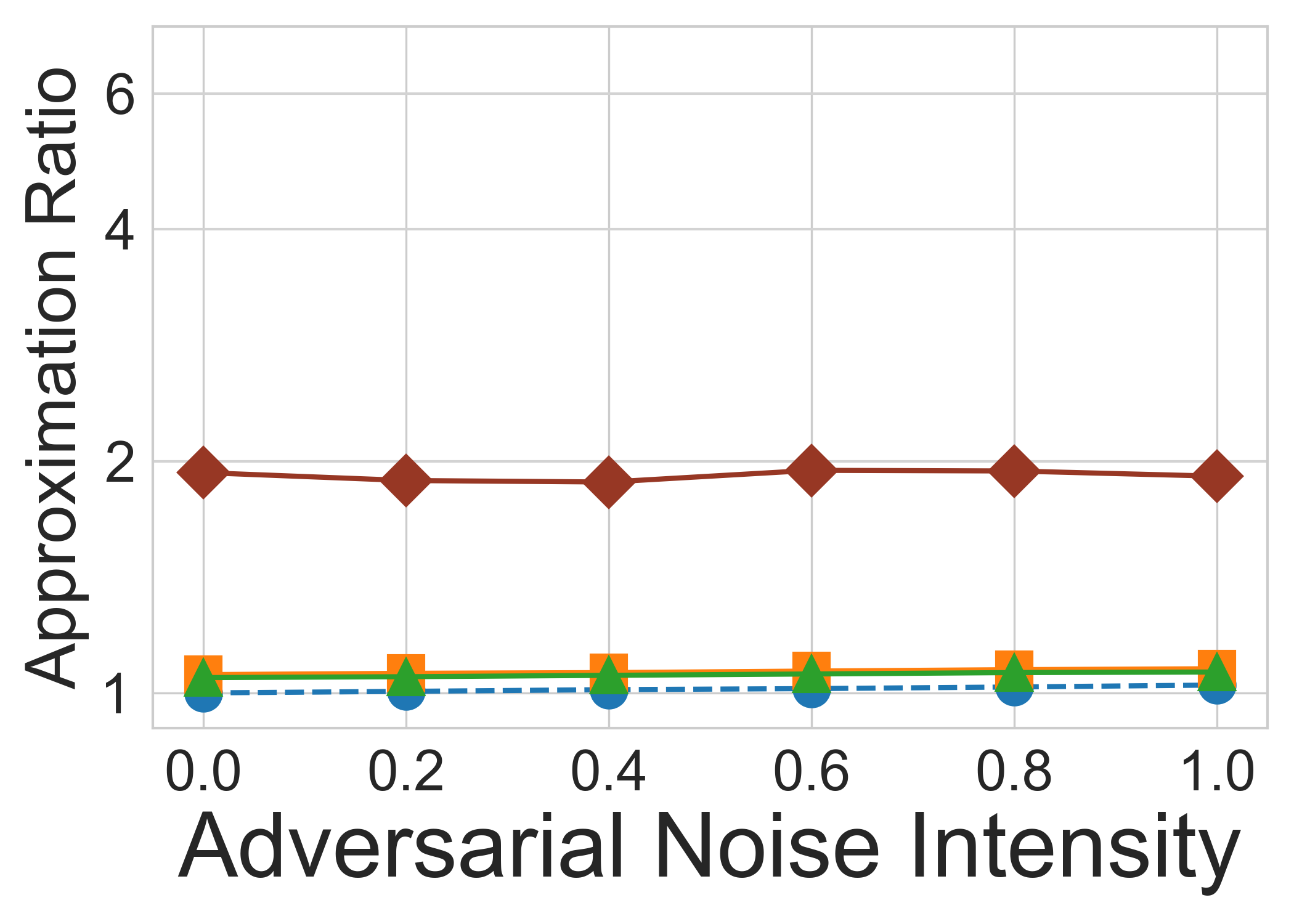}
        \caption{$\alpha = 2.0$}
    \end{subfigure}
    
\caption{\textbf{Smoothness against Stale Noise} (Pareto demand, $n=64$): average communication cost relative to uOPT.}
\label{fig:algs-stale-costs}
\vspace{-3mm}
\end{figure*}

\begin{figure*}[t]
    \centering
    \begin{subfigure}{1\linewidth}
        \centering
        \includegraphics[width=0.35\linewidth]{plots/legend.png}
        \vspace{-2mm} 
    \end{subfigure}
    
    \begin{subfigure}{0.19\linewidth}
        \centering
        \includegraphics[width=1\linewidth]{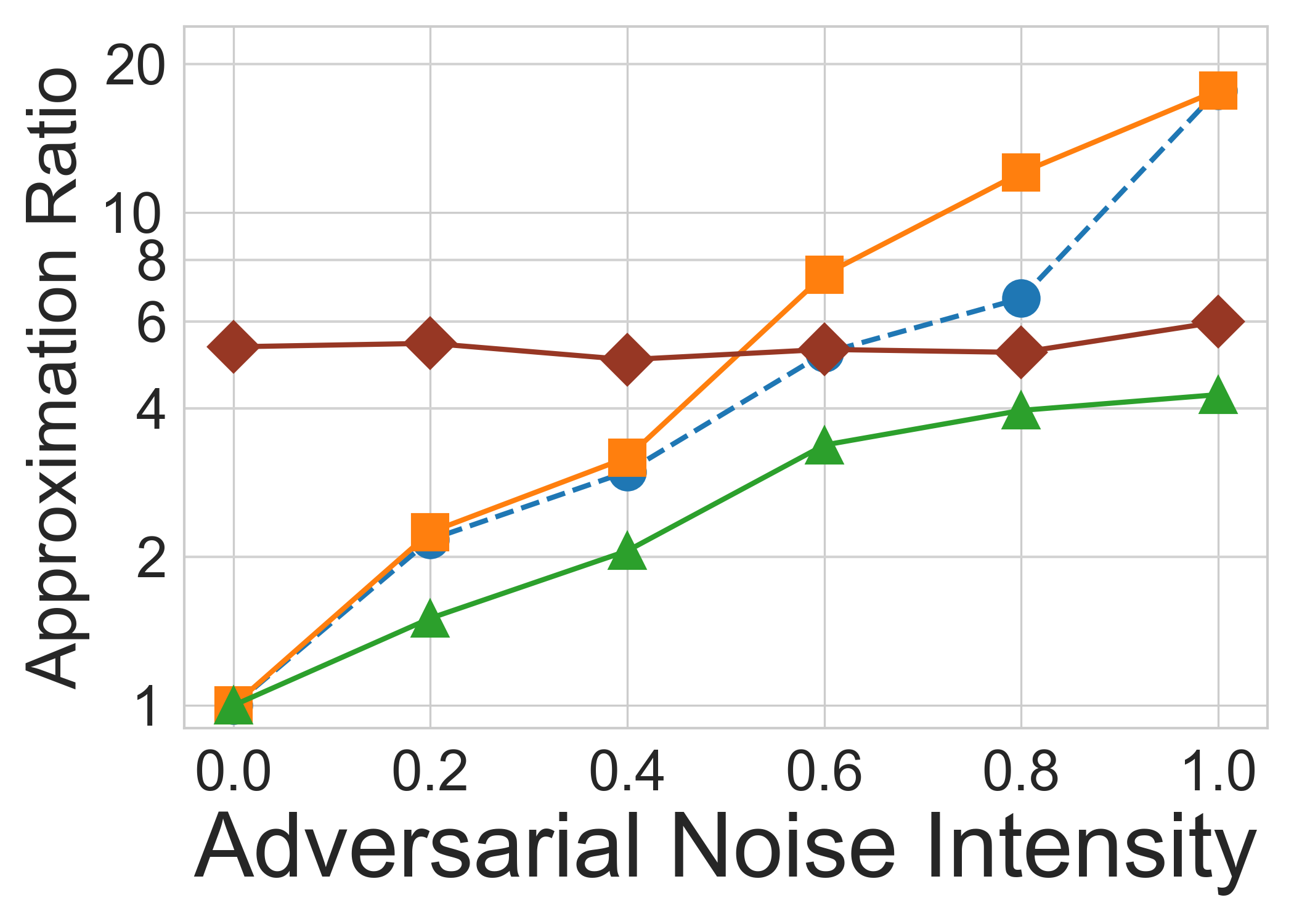}
        \caption{$\alpha = 0.2$}
    \end{subfigure}\hfill
    \begin{subfigure}{0.19\linewidth}
        \centering
        \includegraphics[width=1\linewidth]{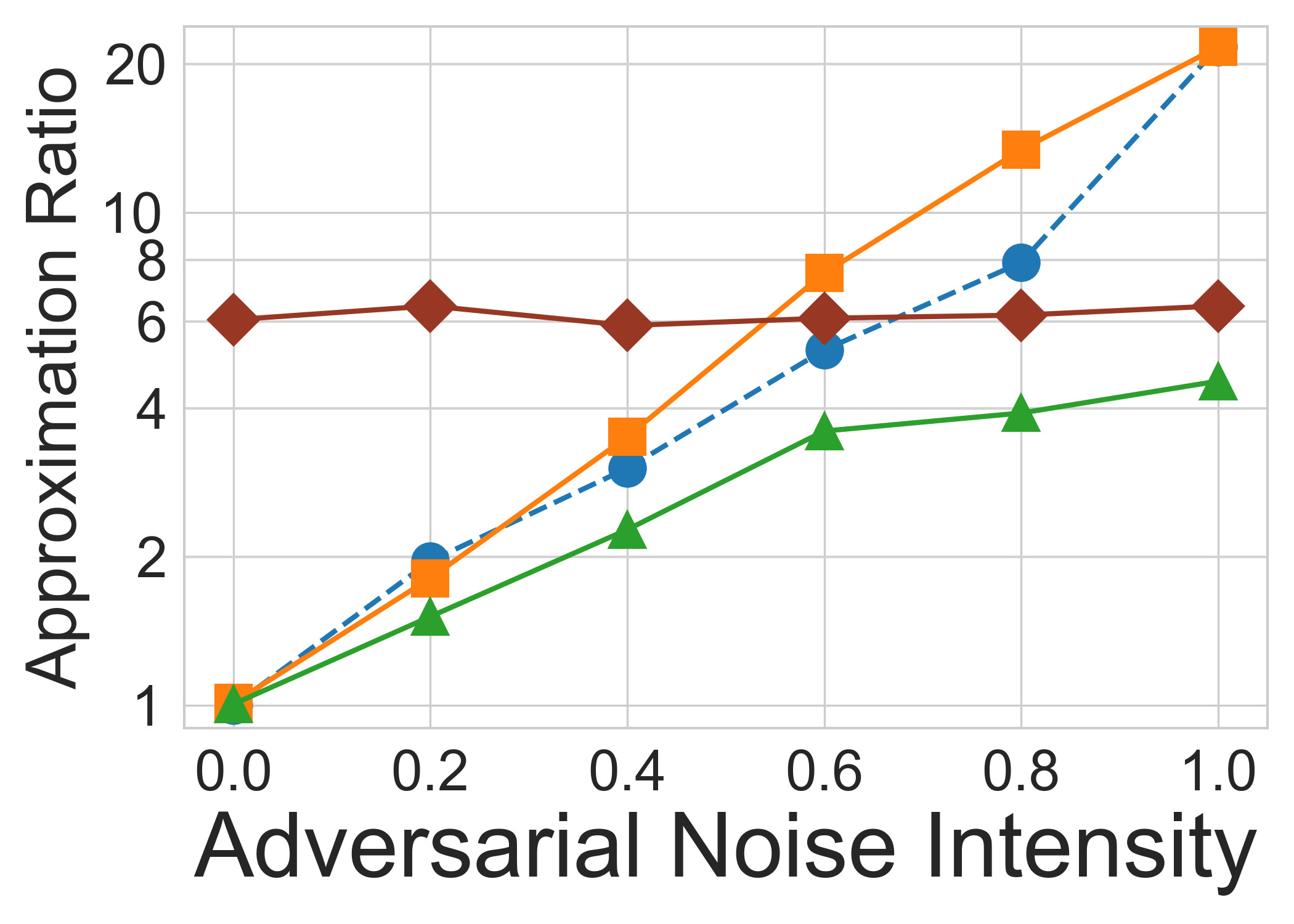}
        \caption{$\alpha = 0.4$}
    \end{subfigure}\hfill
    \begin{subfigure}{0.19\linewidth}
        \centering
        \includegraphics[width=1\linewidth]{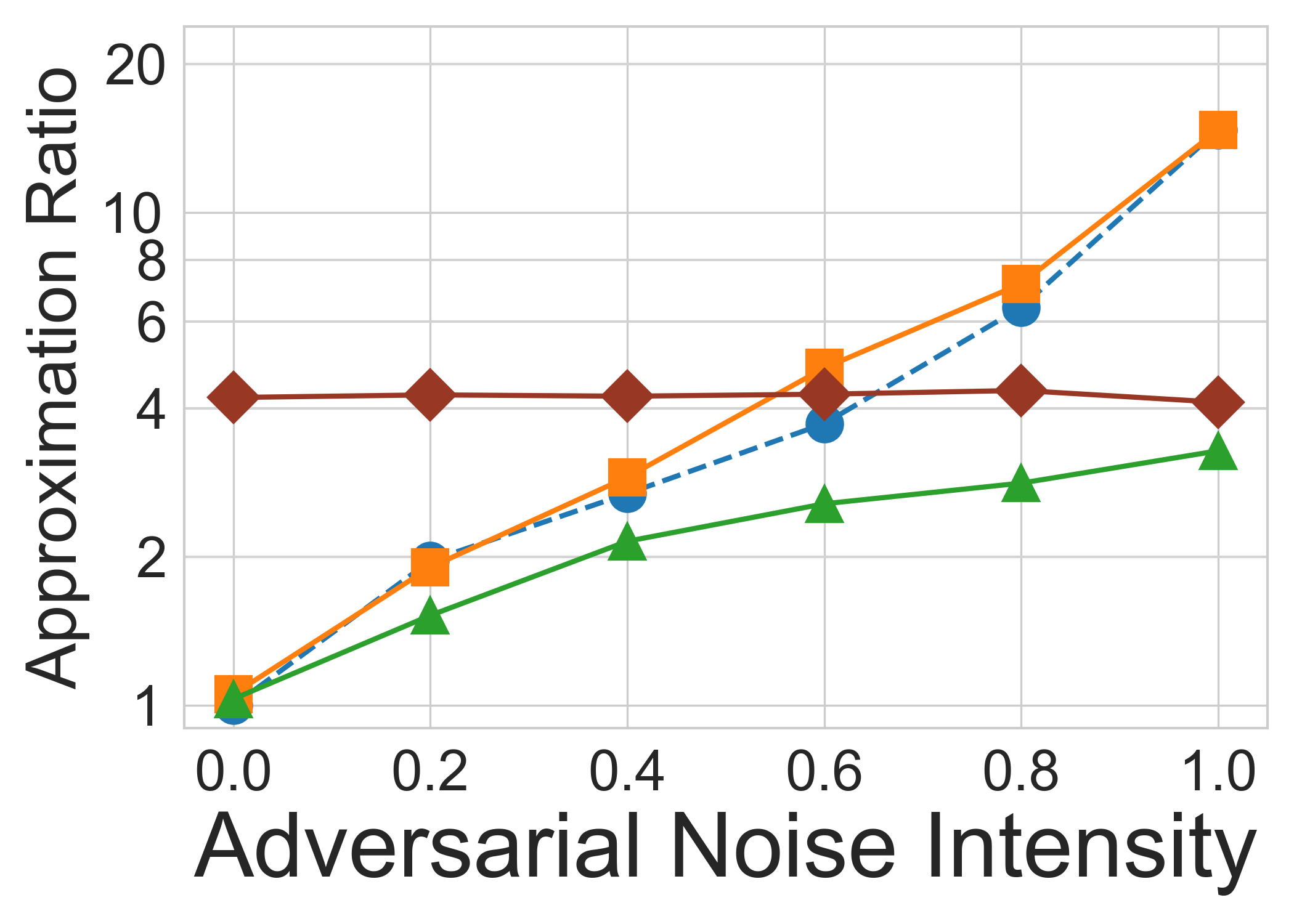}
        \caption{$\alpha = 0.6$}
    \end{subfigure}\hfill
    \begin{subfigure}{0.19\linewidth}
        \centering
        \includegraphics[width=1\linewidth]{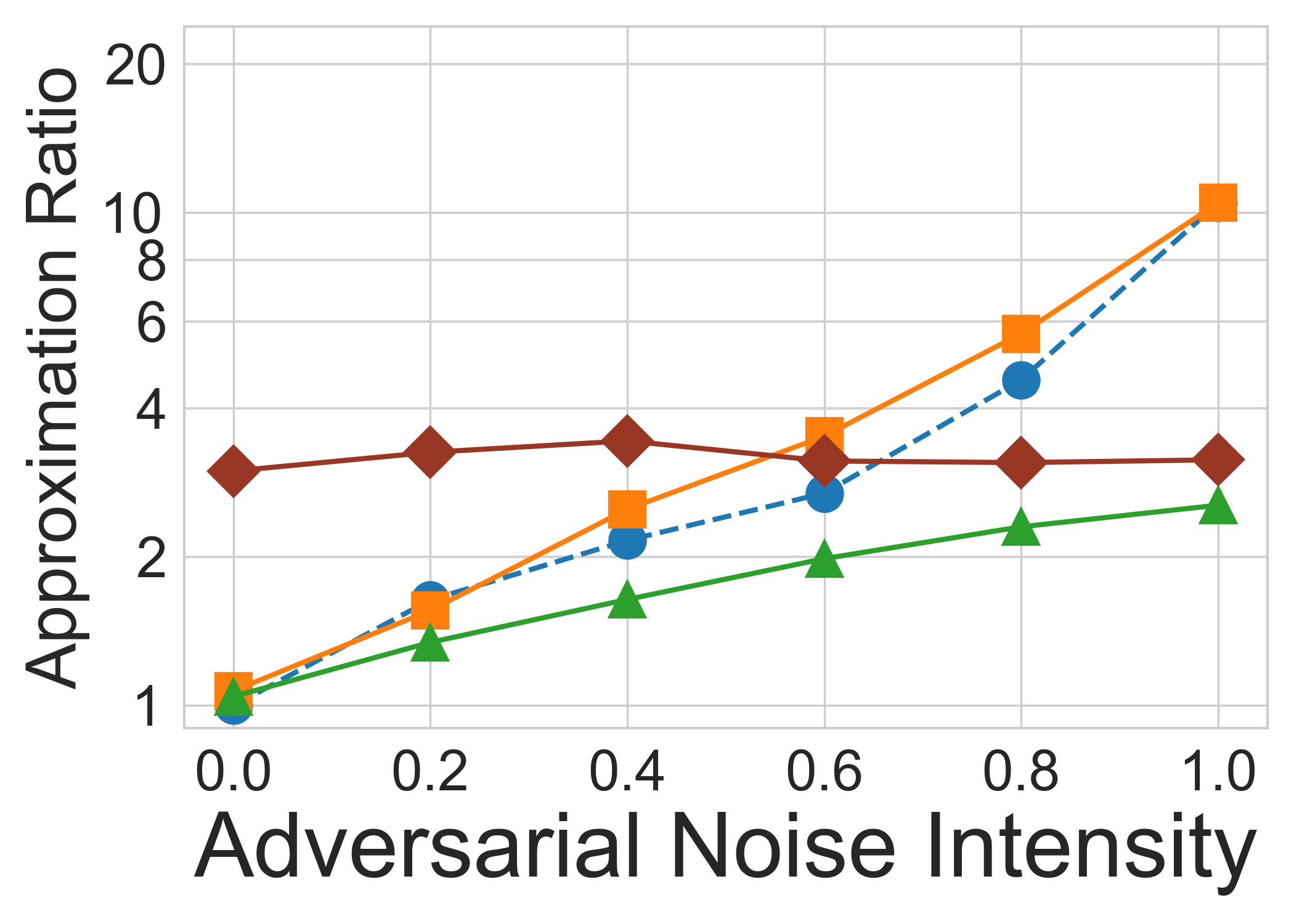}
        \caption{$\alpha = 0.8$}
    \end{subfigure}\hfill
    \begin{subfigure}{0.19\linewidth}
        \centering
        \includegraphics[width=1\linewidth]{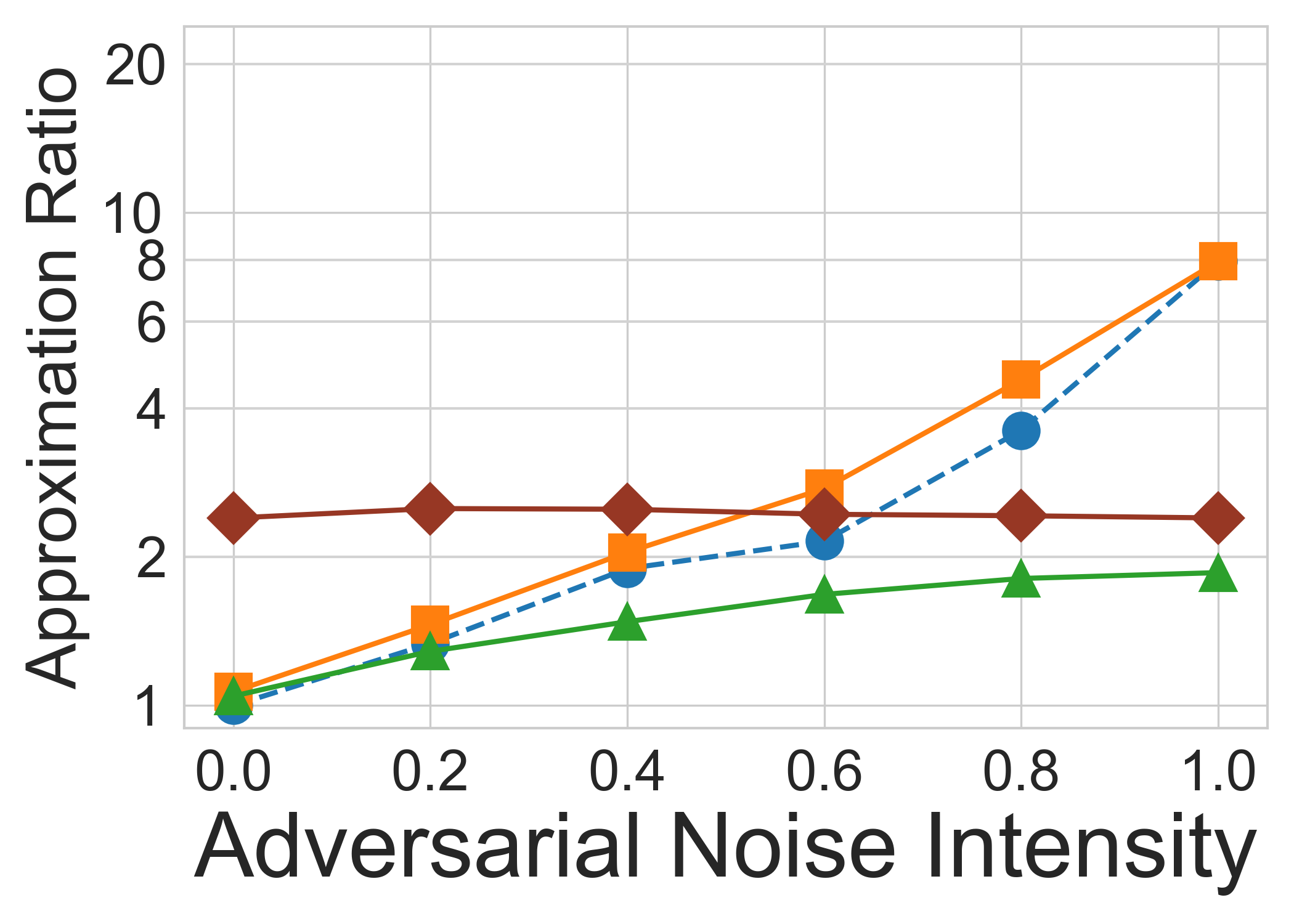}
        \caption{$\alpha = 1.0$}
    \end{subfigure}
    
    \vspace{2mm} 

    \begin{subfigure}{0.19\linewidth}
        \centering
        \includegraphics[width=1\linewidth]{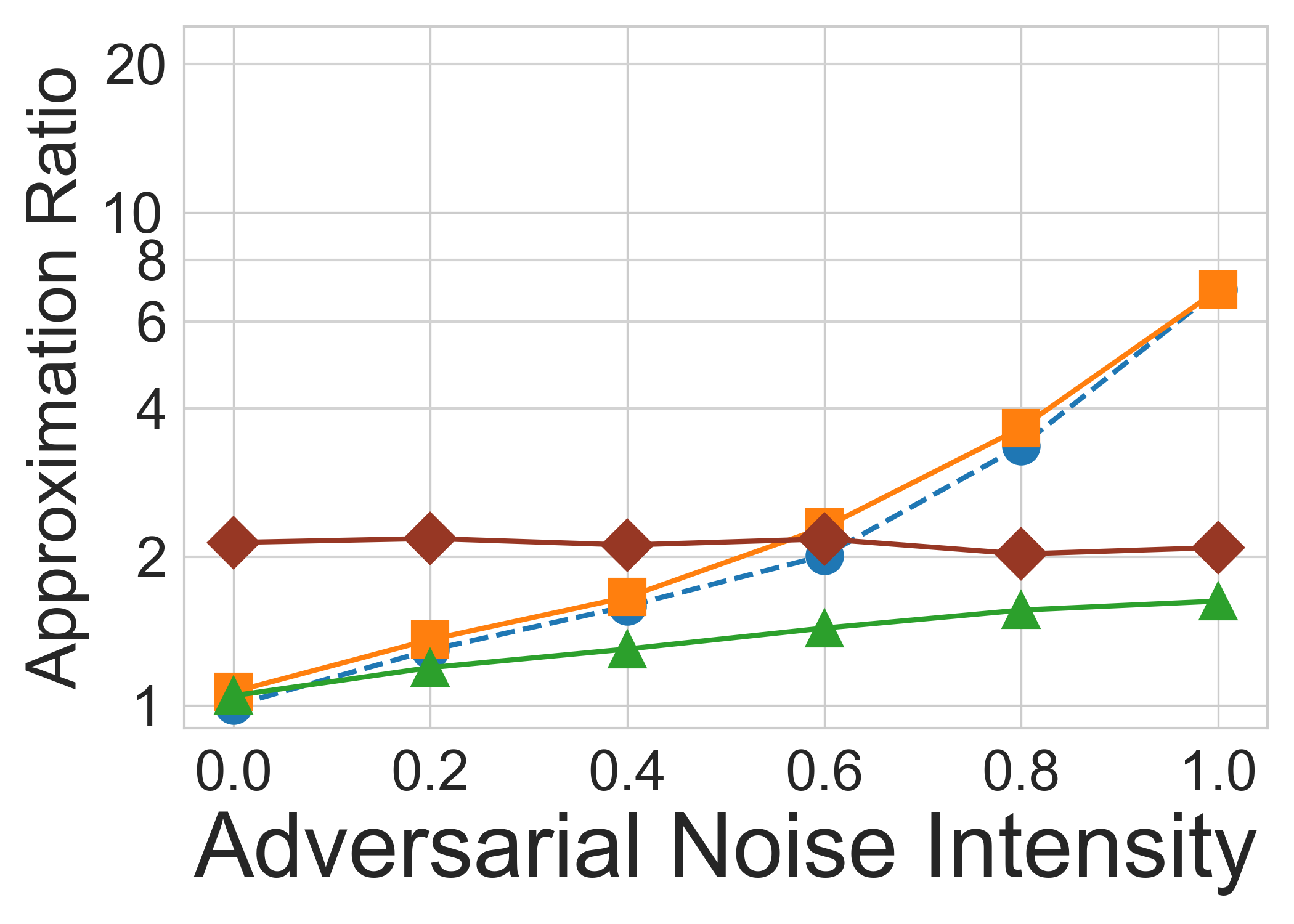}
        \caption{$\alpha = 1.2$}
    \end{subfigure}\hfill
    \begin{subfigure}{0.19\linewidth}
        \centering
        \includegraphics[width=1\linewidth]{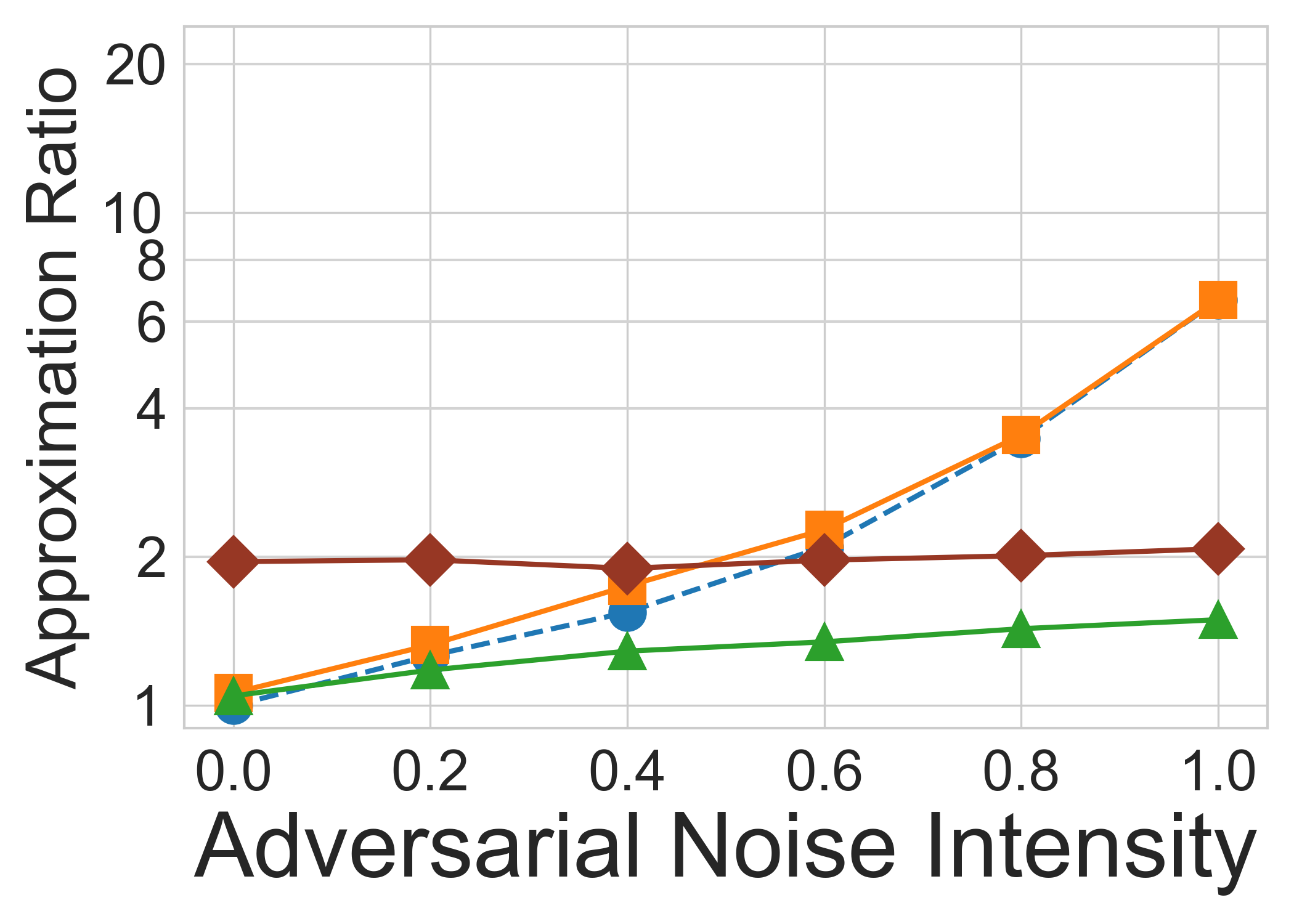}
        \caption{$\alpha = 1.4$}
    \end{subfigure}\hfill
    \begin{subfigure}{0.19\linewidth}
        \centering
        \includegraphics[width=1\linewidth]{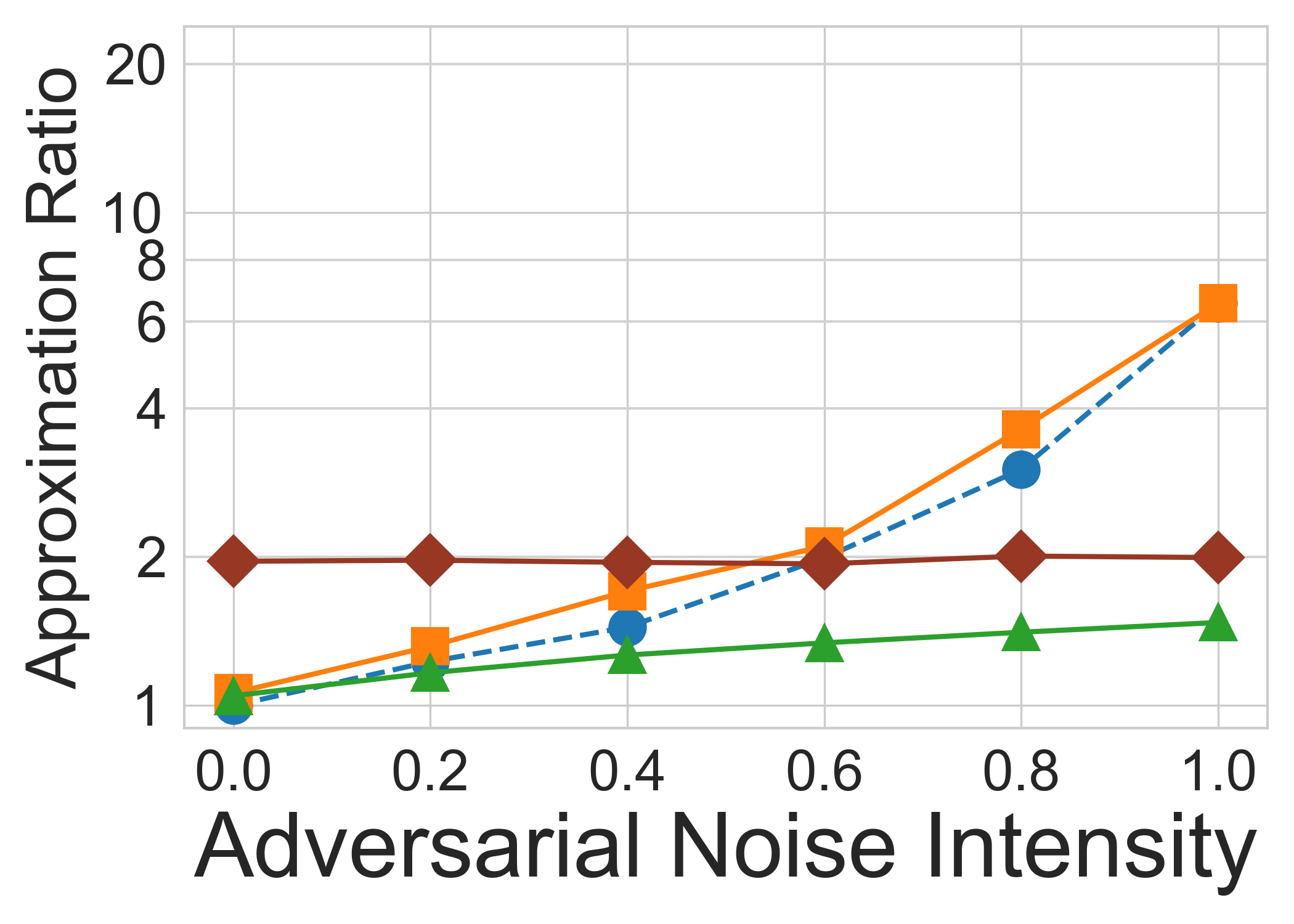} 
        \caption{$\alpha = 1.6$}
    \end{subfigure}\hfill
    \begin{subfigure}{0.19\linewidth}
        \centering
        \includegraphics[width=1\linewidth]{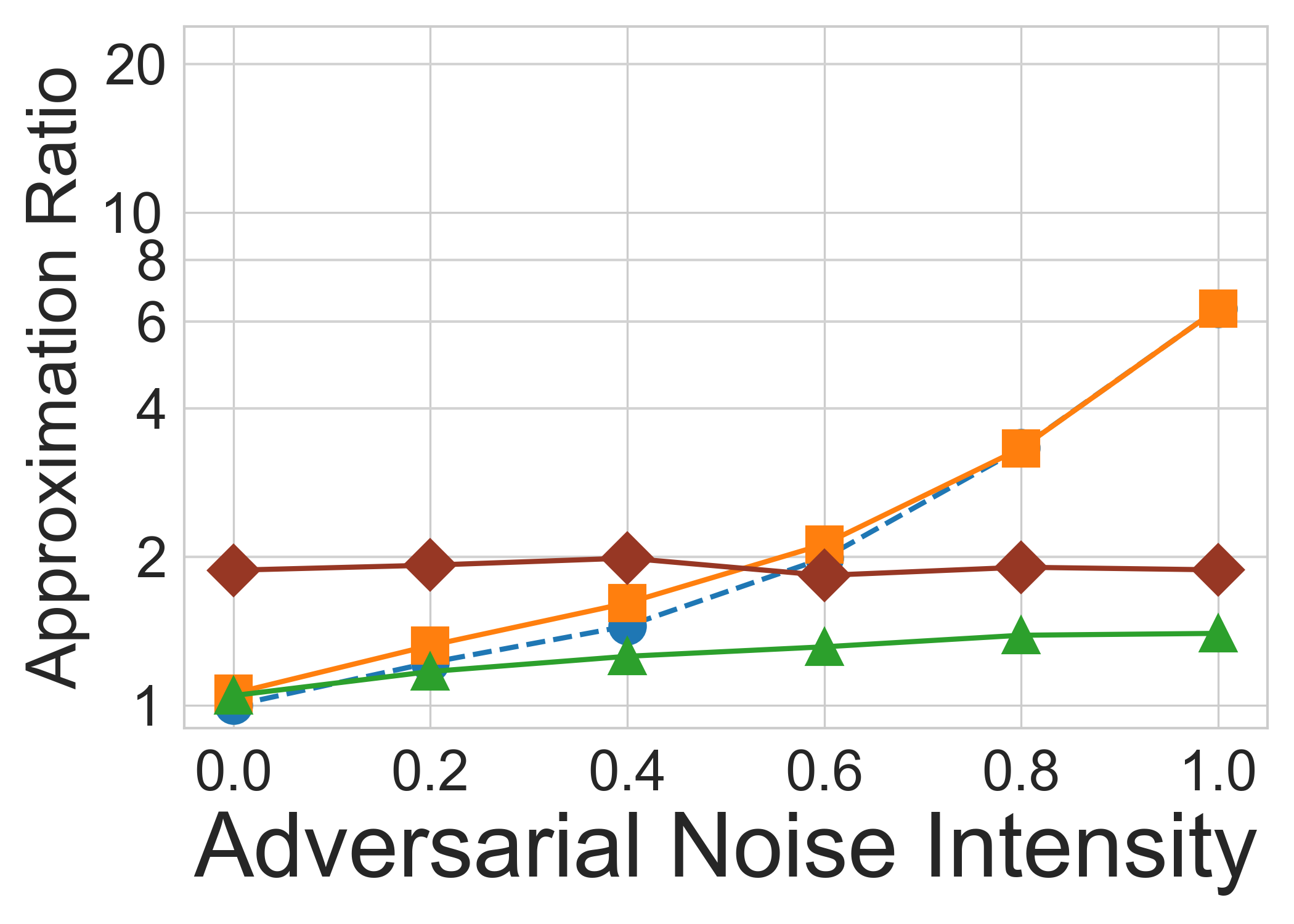}
        \caption{$\alpha = 1.8$}
    \end{subfigure}\hfill
    \begin{subfigure}{0.19\linewidth}
        \centering
        \includegraphics[width=1\linewidth]{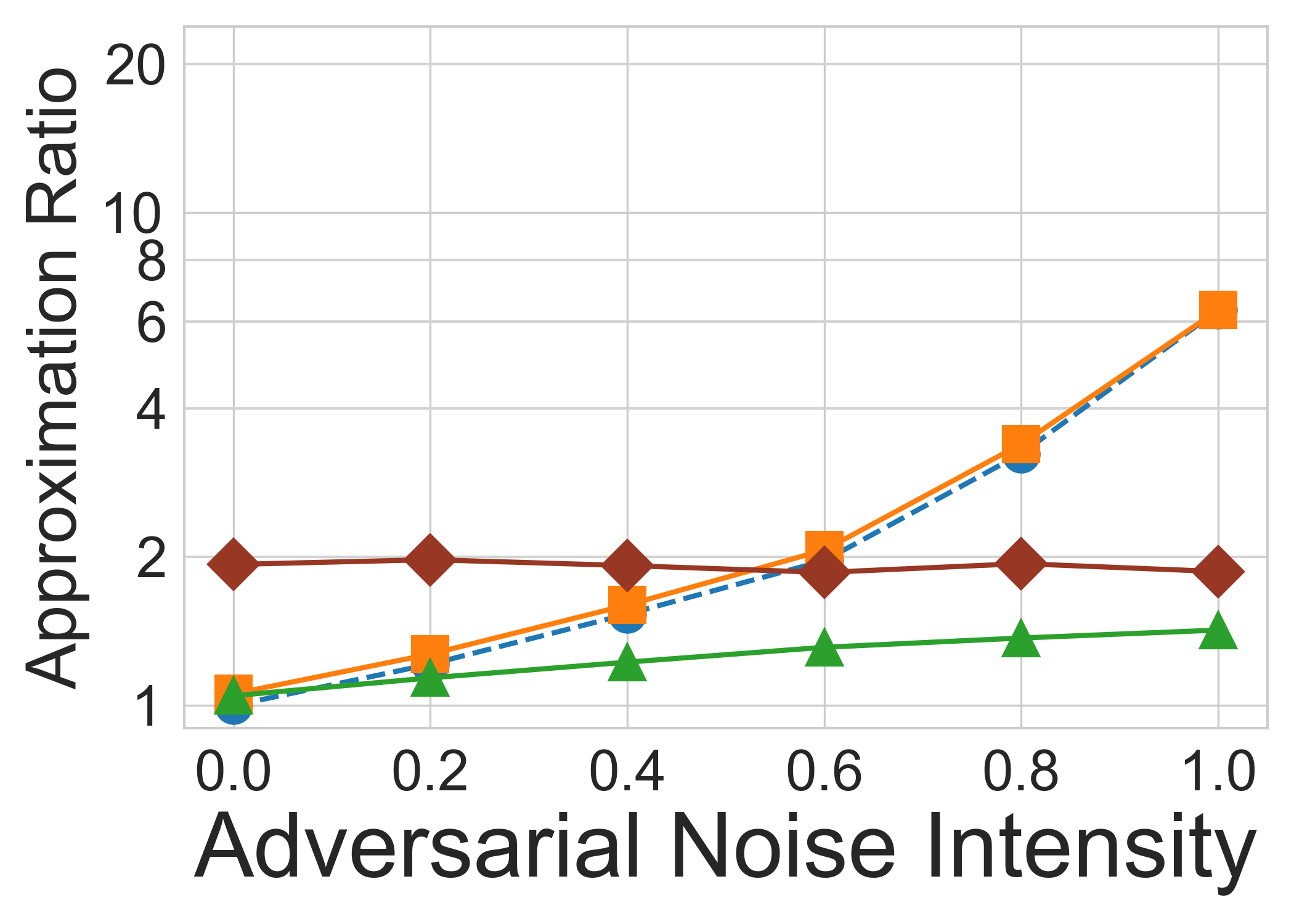}
        \caption{$\alpha = 2.0$}
    \end{subfigure}
    
\caption{\textbf{Smoothness against Adversarial Noise} (Pareto demand, $n=64$): average communication cost, relative to uOPT.}
\label{fig:algs-adversarial-costs}
\vspace{-3mm}
\end{figure*}

\begin{figure*}[t]
    \centering
\begin{subfigure}{1\linewidth}
  \hspace{9.5cm}
\includegraphics[width=0.33\linewidth]{plots/legend.png}
\vspace{-2mm} 
\end{subfigure}
\begin{subfigure}{0.35\linewidth}
\centering
\includegraphics[width=\linewidth]{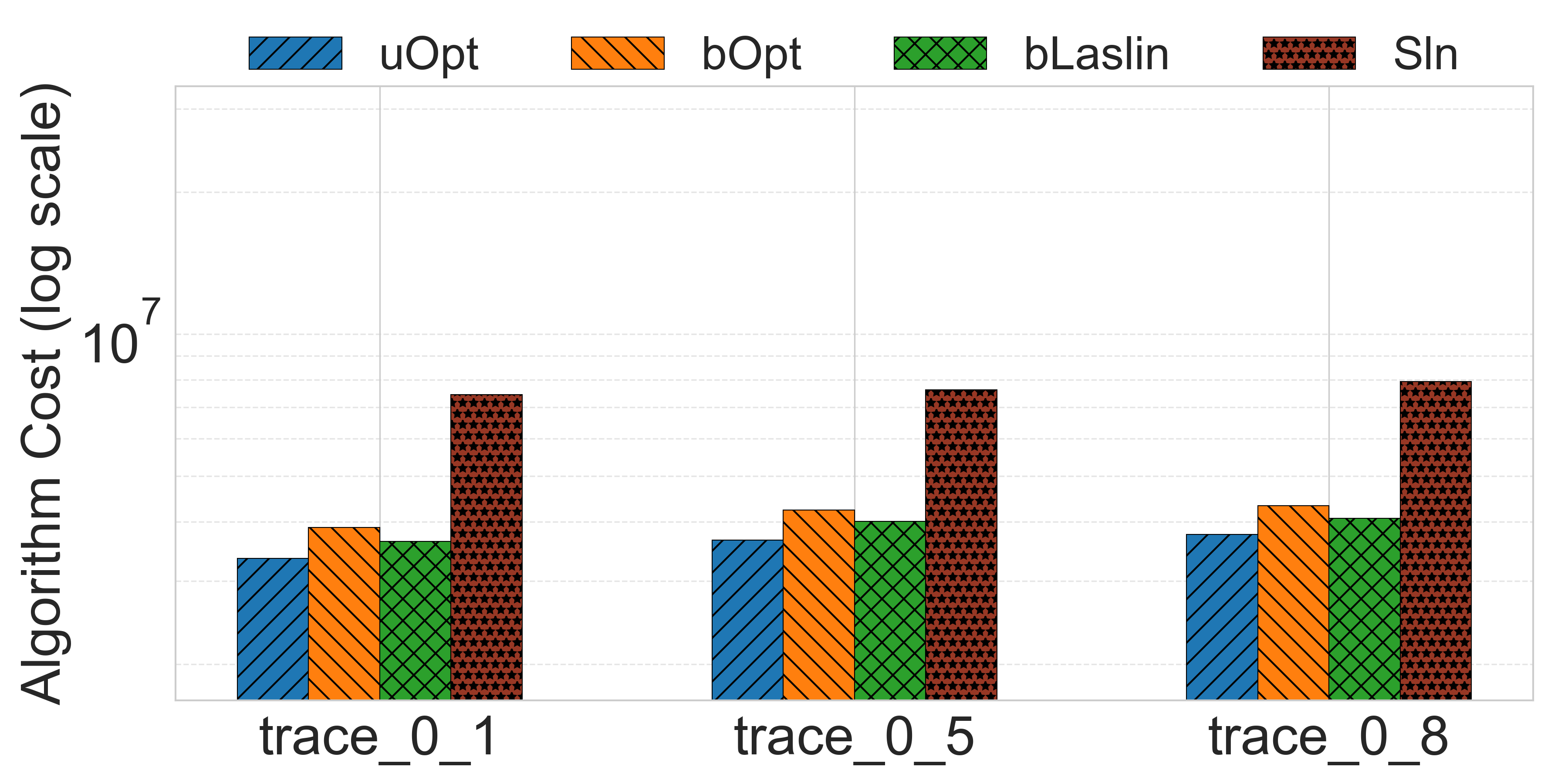}
\caption{pFabric traces ($\eta=0$)}
\label{fig:result-pfabric-consistency}
\end{subfigure}\hfill
\begin{subfigure}{0.21\linewidth}
\centering
\includegraphics[width=1\linewidth]{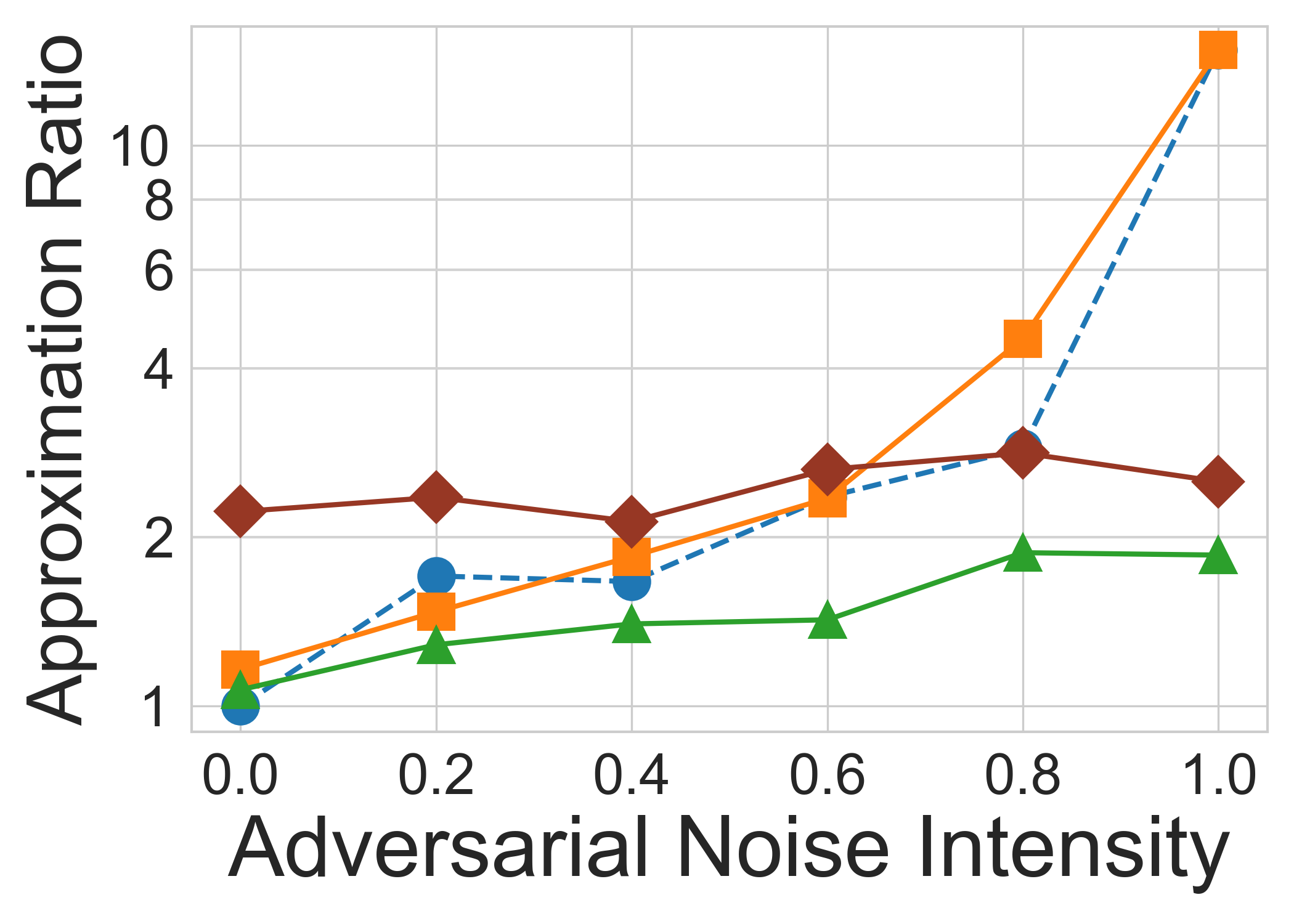}
\caption{pFabric trace$\_0\_1$}
\end{subfigure}\hfill
\begin{subfigure}{0.21\linewidth}
\centering
\includegraphics[width=1\linewidth]{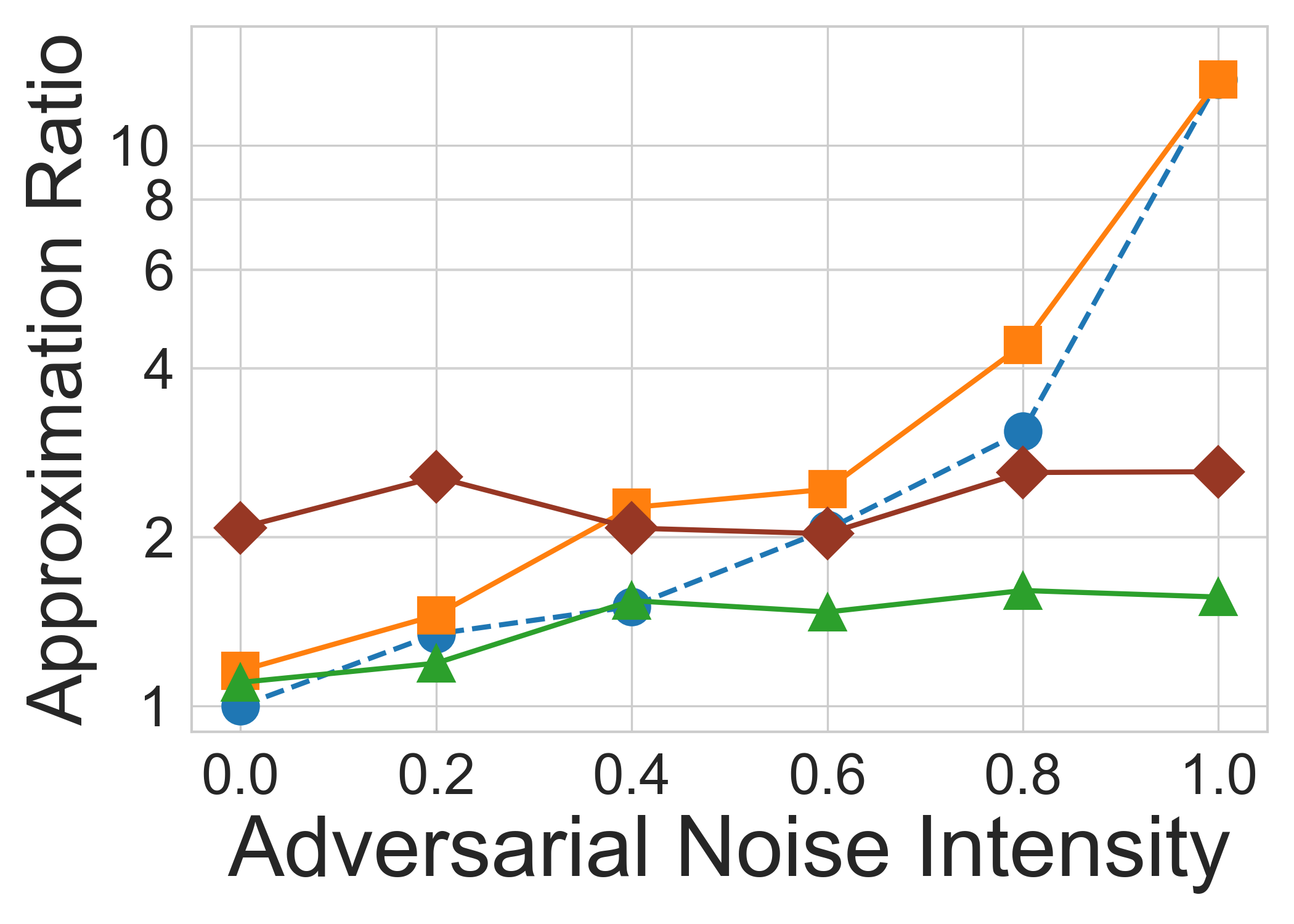}
\caption{pFabric trace$\_0\_5$}
\end{subfigure}\hfill
\begin{subfigure}{0.21\linewidth}
\centering
\includegraphics[width=1\linewidth]{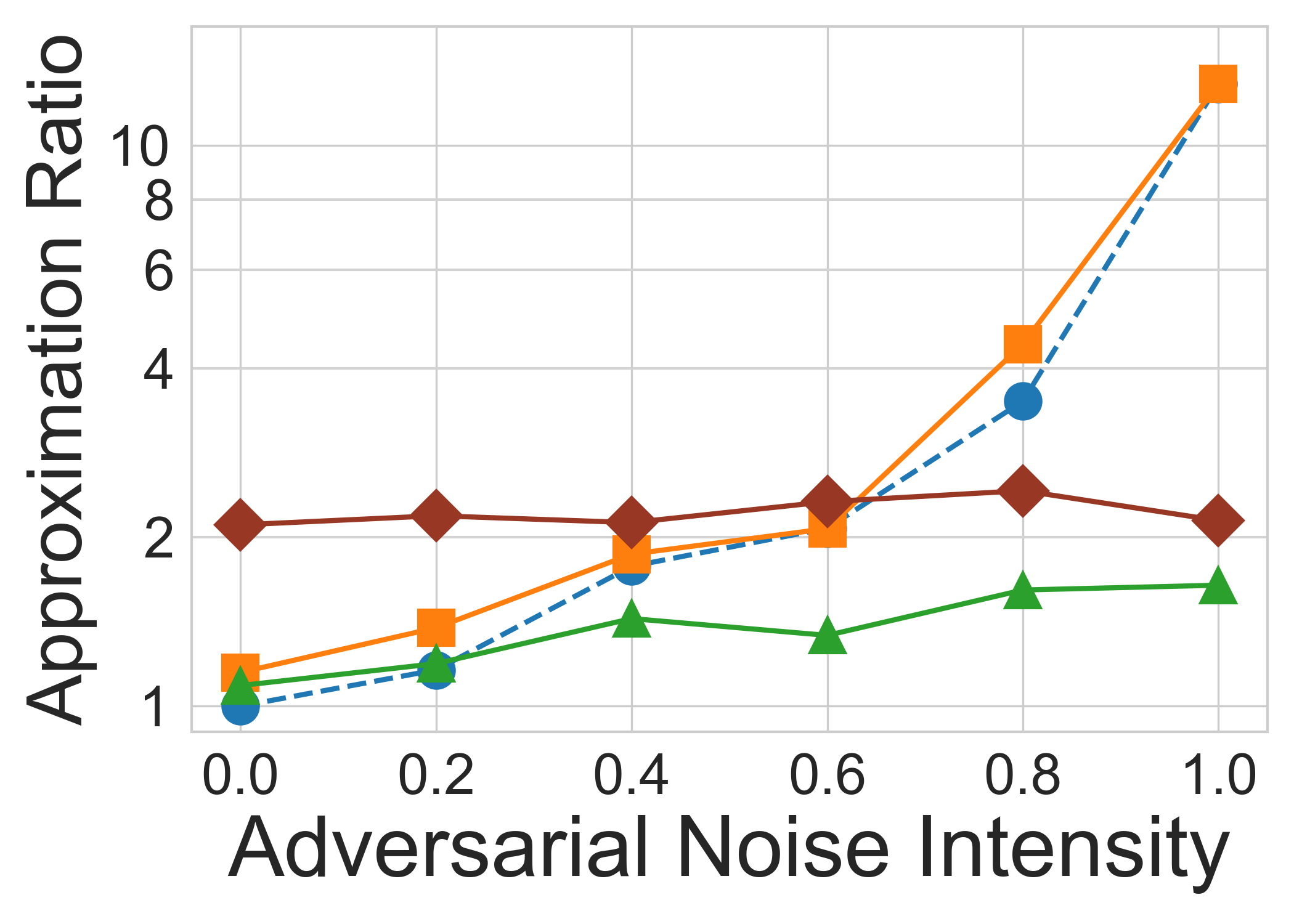}
\caption{pFabric trace$\_0\_8$}
\end{subfigure}  
\caption{ \textbf{pFabric demand under Adversarial Noise} ($n=144$): consistency (a) and smoothness (b),(c),(d).}
\label{fig:algs-pfab-adversarial-costs}
    \vspace{-3mm}
\end{figure*}

\begin{figure}[t]
\begin{subfigure}{1\linewidth}
\centering
\includegraphics[width=0.65\linewidth]{plots/legend.png}
\end{subfigure}
\begin{subfigure}{0.495\linewidth}
\centering
\includegraphics[width=\linewidth]{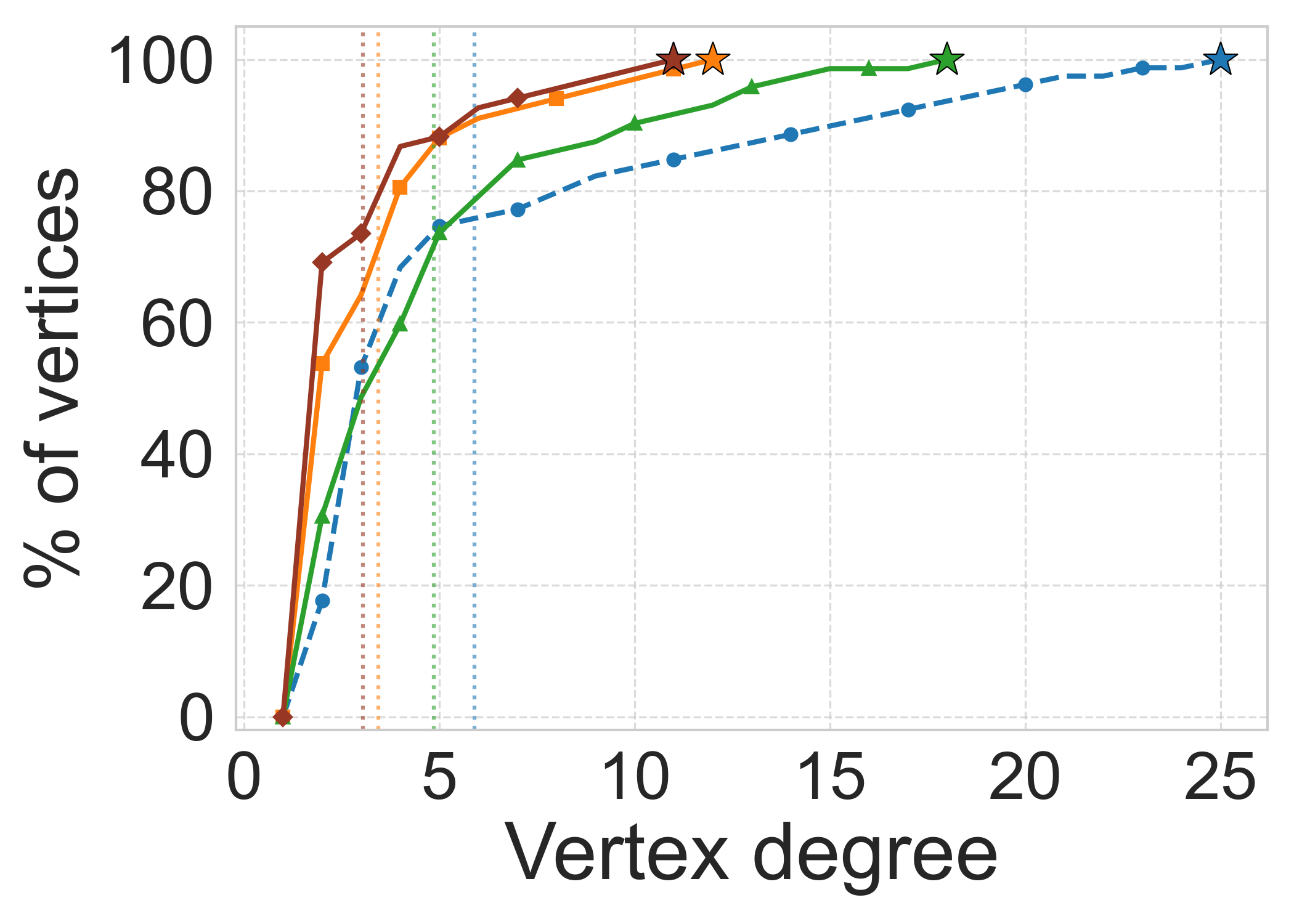}
\caption{Pareto $\alpha=0.2, n=64$}
\label{fig:result-pareto-vertex-cdf}
\end{subfigure}\hfill
\begin{subfigure}{0.495\linewidth}
\centering
\includegraphics[width=1\linewidth]{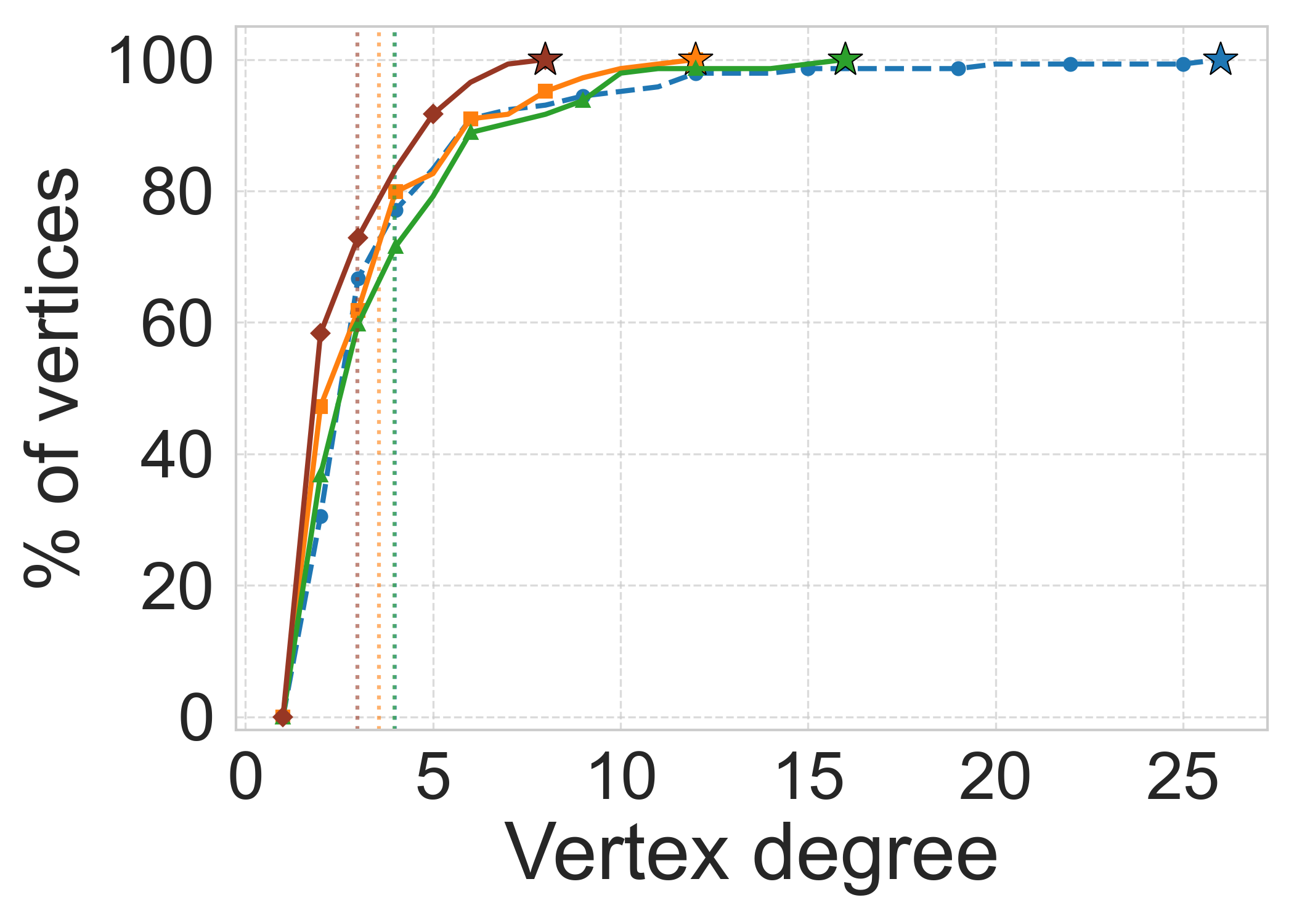}
\caption{pFabric trace$\_0\_5, n=144$}
\label{fig:result-pfabric-vertex-cdf}
\end{subfigure}
\caption{\textbf{Node degree CDF} ($\eta=0$): the vertical lines indicate the average node degree. 
} 
\label{fig:result-vertex-cdf}
\vspace{-3mm}
\end{figure}

\section{Experimental Evaluation}
\label{sec:experiments}

In this section, we present our experimental setup and results.


\paragraph*{Algorithms and baselines.} To evaluate the performance of \laslin, we implemented the following algorithms:
\begin{itemize}
\item \textbf{Sln:} Demand-Oblivious Discrete \sln (\autoref{def:sln}, with the height of each node having been assigned the number of consecutive independent trials in which a random variable $x \in [0, 1]$ satisfied $x < 0.5$); 
\item \textbf{uOPT:} Unbounded Optimum Static \sln  (\autoref{def:ossln}) on predicted demand;
\item \textbf{bOPT:} Bounded Optimum Static \sln (\autoref{def:ossln-bounded} on predicted demand, with $h_{\max} = 6$); 
\item \textbf{bLasin}: \laslin (\autoref{def:learning-augmented}, with bOPT as \textit{input prediction});
\end{itemize}

\paragraph*{Demand matrices.}
In our experiments, we used both synthetic and real-world workloads. The synthetic workloads were based on a Pareto distribution, with the shape parameter $\alpha \in [0.2,2.0]$ (the higher the value of $\alpha$, the closer the demand to a uniform distribution). For each value of $\alpha$, we generated $30$ distinct demand matrices on a network with $64$ nodes. To validate our algorithms on real-world data, we utilized three pFabric traces \cite{pfabric}, generated by executing simulation scripts in NS2 on a network with $144$ nodes. 

\paragraph*{Noise generation.}
In line with some previous related work ~\cite{LykourisV18}, we evaluate the smoothness of our algorithms by introducing noise, or \textit{prediction error}, either to the demand weight values (between pairs of nodes) or to the heights assigned to nodes by an \sln algorithm.
In particular, we implemented two noise models:
\begin{itemize}
    \item \textbf{Stale Noise:} This model aims to simulate the effect of an \textit{outdated input prediction} through a probabilistic swapping mechanism. For every node pair in the demand matrix, the  \textit{demand weight value} is replaced by a different value, taken from a historical or reference workload, with probability $\eta$ (the noise intensity parameter).
    \item \textbf{Adversarial Noise:} This model simulates the effect of the presence of \textit{malicious nodes} in the network. For every network node, the \textit{height value} assigned to it by the Skip List algorithm is reset to $1$, with probability $\eta$.
\end{itemize}

The Stale Noise model represents scenarios with delayed updates of input prediction values, caused by, e.g., network reconfigurations or synchronization errors. We restricted the evaluation of smoothness against stale noise to demand matrices based on Pareto distributions.\footnote{We excluded pFabric traces from this part of the evaluation, as the dataset lacked sufficient independent traces to generate reliable reference models.} Given a workload $W_i, 1\leq i< 30$, we designated the subsequent workload, $W_{i+1}$, as the reference source for the current workload, $W_i$. Specifically, for each communication pair $(x, y)$, the actual demand weight $W_i(x, y)$ was replaced by the reference weight $W_{i+1}(x, y)$, with probability $\eta$.

In the Adversarial Noise model, adversarial agents reset their own node heights to a baseline value of $1$, with probability $\eta$, eliminating structural (\sln) shortcuts and increasing the total routing work (communication cost). The objective of this experiment is to demonstrate the efficacy of a simple countermeasure: by increasing each node's height value with a small randomized value, we can effectively safeguard the network’s performance against such adversarial intent.


\paragraph*{Metrics.} We used the following performance metrics:
\begin{itemize}
    \item Communication cost (\autoref{def:cost}); 
    \item Consistency (\autoref{def:consistency-robustness}) and smoothness;
    \item Node degree cumulative distribution function (CDF): percentage of nodes with degree $\leq x, x\in[1,n)$, and the average node degree.
\end{itemize}

\paragraph*{Results.} 
In Figures \ref{fig:result-pareto-consistency} and \ref{fig:result-pfabric-consistency}, we evaluate the consistency of \laslin under Pareto-distributed demand with variable shape parameter ($\alpha\in[0.2,2.0]$), assuming perfect prediction advice ($\eta=0$). We compare the (average) total communication cost of \laslin (bLaslin) against three baselines: the bounded static optimum (bOPT), the unbounded static optimum (uOPT), and the demand-oblivious discrete \sln (Sln). The results demonstrate that, unlike the demand-oblivious baseline, \laslin achieves a total communication cost close to the optimum across all simulated workloads, regardless of spatial locality ($\alpha$). This empirically confirms that \laslin is consistent.

Notably, \laslin, despite being projected onto the bOPT structure (bLaslin), i.e., using the solution of bOPT as its \textit{input prediction}, actually achieves better communication cost than bOPT, even in noiseless scenarios (Fig. \ref{fig:result-pareto-consistency}). This performance gain is enabled by the slight increase in the network's average and maximum degrees, as shown in Figures \ref{fig:result-pareto-vertex-cdf} and \ref{fig:result-pfabric-vertex-cdf}. This behavior has been expected, since bLaslin effectively breaks height ties, inherent in the bOPT solution. By adding a small random remainder to the predicted heights, nodes that have previously shared the same rank (height) become differentiated, generating additional (\sln) shortcut links that further reduce routing cost.

In Figure \ref{fig:algs-stale-costs}, we evaluate the smoothness of \laslin against increasing levels of stale noise. Both uOPT and bOPT are significantly impacted by the intensity of the stale noise (increasing $\eta$ on the x-axis), especially for lower $\alpha$ values, where the workload is more skewed. In contrast, bLaslin consistently achieves lower communication cost than its input prediction (bOPT). This demonstrates that \laslin remains robust relative to both the type and magnitude of input prediction errors.

Additionally, we observe that, as $\alpha$ increases, approaching a uniform demand distribution, the performance of the optimum \sln converges to that of the discrete demand-oblivious \sln.

The strongest advantage of \laslin is observed in the presence of adversarial noise (Figures \ref{fig:algs-adversarial-costs} and \ref{fig:algs-pfab-adversarial-costs}(b,c,d), note that the y-axis of these plots is shown in the log scale). 
By simply projecting the \laslin mechanism onto the predicted node heights, we successfully mitigate the disruption that could be caused by adversarial agents. We observe that while both uOPT and bOPT degrade progressively as noise intensity increases, eventually underperforming even the demand-oblivious discrete \sln baseline, bLaslin consistently achieves lower communication costs than both its reference prediction (bOPT) and the lower bound baseline (uOPT).

Finally, it is important to note that the average and the maximum node degree values obtained in the simulations (Fig. \ref{fig:result-vertex-cdf}) are in alignment with the theoretical expectation (\autoref{thm:expected-degree}, \autoref{thm:degree-whp}
). In particular, for the discrete \sln, the maximum degree is approximately $2 \cdot\log n$ (around $12$ for $n=64$), which, in the bounded case, was exactly equal to the limit of $2\cdot h_{\max} = 12$. Notably, the average degree remained tightly consistent across all algorithms. This confirms that the degree distribution of an \sln is heavily skewed, with the vast majority of nodes maintaining a low degree.

\section{Related Work}

\noindent \textbf{Improved designs of overlay networks.} Overlay networks have been the focus of computer scientists for a long while (see~\cite{LuaCPSL05} 
for a detailed survey). 
Many overlay networks of constant epected degree and logarithmic paths lengths are known, for example~\cite{continuous-discrete03} that extends the work of~\cite{FraigniaudG06}, giving a general overlay construction such that for a network with degree $d$ and $n$ nodes, it is possible to guarantee path length of $\Theta(\log_d n)$.
Lately, improving the design of networks given prior information about communication demand has attracted the attention of researchers. In particular, the work of~\cite{avin2018toward} set the scene for demand-aware network design, which was then explored more with particular focus on datacenter networking~\cite{SouzaGS22,SeedTree,SpiderDAN}.  

Works closest to ours are by~\cite{SkipListNetwork,CirianiFLM07} in which the authors discuss a self-adjusting variant of skip list network: one that adjusts over time based on the incoming demand, which is in contrast with our work that focuses on an optimal skip list network with static height given the communication demand as the input. This is particularly important, as self-adjustments can introduce additional cost (e.g., in terms of delay) to the system.
Also, the work of~\cite{MandalCK15} discussed how a skip-list construction can be used for peer-to-peer communications, however they do not provide theoretical upper bounds for the cost of their overlay network.
Lastly, we point out to other extensions of skip list, namely Skip Graph~\cite{AspnesS03} and SkipNet~\cite{HarveyJSTW03} that, to facilitate overlay communications, consider expected degree to be $O(\log n)$, in contrast to our design, which has a constant degree in expectation.

\noindent \textbf{Learning-augmented algorithms.} 
Given the rise of machine learning algorithms (which inherently comes with some errors), and after the pioneering work of~\cite{PurohitSK18} on ski rental \& job scheduling, the focus has been shifted on studying consistency and robustness of algorithms augmented with the learned advice.
Such an approach has been used in various applications so far\footnote[2]{See a full list of recent works in the following website: \url{algorithms-with-predictions.github.io}}, for example caching~\cite{LykourisV18},
online TSP~\cite{GouleakisLS23},
list update~\cite{AzarLS25}.

Some of the more recent works focus on improving single-source search data structures. Initially, the works of~\cite{ST1,ST2} provided learning-augmented search trees. The works that are closest to our work are by~\cite{MLSkipList1,MLSkipList2}, which studied learning-augmented single-source skip lists. In contrast, we consider a more challenging model, where the focus is on pairwise communications rather than communication from a single source.
Recently, we became aware of a not peer-reviewed manuscript~\cite{P2PWithAdvice} that aims to augment overlay networks with learned advice; however, this paper looks into improving self-stabilization and recovery of overlay networks, not improving communication delay.


\bibliographystyle{named}
\bibliography{reference}

\section{Appendix}

\begin{algorithm}[h]

\caption{find\_neighbors($v$, $h_v$, current\_height)}\label{alg:find_neighbors}
\tcp{code running at node $u$}
\If{current\_height = $- \infty$}{
  $\text{anchor}_u \gets v$
}

$N \gets \emptyset$\\
\If{current\_height $< h_u$}{
  $N \gets \{u\}$\\
  $N_u \gets N_u \cup \{v\}$\\
  Update $D_u$
}
\If{$u < v$}{
  $N \gets \{x \in N_u | x \le v \text{ or } (h_v < D_u[x] \text{ and } h_v < h_u)\}$\\
  \For{$x \in N_u$}{
    \If{$x < u \text{ and } h_u \le D_u[x] \le h_v$}{
      $N \gets N \cup x.\text{find\_neighbors}(v, h_v, h_u)$
    }
  }
}

\Else{
  $N \gets \{x \in N_u | v \le x \text{ or } (h_v < D_u[x] \text{ and } h_v < h_u)\}$\\
  \For{$x \in N_u$}{
    \If{$u < x \text{ and } h_u \le D_u[x] \le h_v$}{
      $N \gets N \cup x.\text{find\_neighbors}(v, h_v, h_u)$
    }
  }
}
\Return $N$

\end{algorithm}

\begin{algorithm}[h]

\caption{delete\_node($v$, $h_v$, current\_height)}\label{alg:delete_neighbors}
\tcp{code running at node $u$}
$N \gets \emptyset$\\
\If{current\_height $< h_u$}{
  $N \gets \{u\}$\\
  $N_u \gets N_u \setminus \{v\}$\\
  Update $D_u$
}
\If{$u < v$}{
  $N \gets \{x \in N_u | x \le v \text{ or } (h_v < D_u[x] \text{ and } h_v < h_u)\}$\\
  \For{$x \in N_u$}{
    \If{$x < u \text{ and } h_u \le D_u[x] \le h_v$}{
      $N \gets N \cup x.\text{delete\_node}(v, h_v, h_u)$
    }
  }
}

\Else{
  $N \gets \{x \in N_u | v \le x \text{ or } (h_v < D_u[x] \text{ and } h_v < h_u)\}$\\
  \For{$x \in N_u$}{
    \If{$u < x \text{ and } h_u \le D_u[x] \le h_v$}{
      $N \gets N \cup x.\text{delete\_node}(v, h_v, h_u)$
    }
  }
}
\Return $N$

\end{algorithm}

\end{document}